\documentclass{article}
\usepackage[utf8]{inputenc}
\pdfoutput = 1
\usepackage{fullpage,amsthm,amsmath,natbib,algorithm,algorithmic,enumitem,afterpage,amssymb,setspace,amsfonts,array,verbatim,commath,newfloat}
\usepackage{graphicx}
\usepackage[dvipsnames]{xcolor}
\usepackage{tikz}
\usetikzlibrary{arrows,matrix,positioning, fit, shapes.geometric}
\usepackage[hang,flushmargin]{footmisc}
\usepackage[colorlinks=TRUE,citecolor=Blue,linkcolor=BrickRed,urlcolor=PineGreen]{hyperref}

\newtheorem{theorem}{Theorem}

\newtheorem{corollary}{Corollary}

\newtheorem{definition}{Definition}

\allowdisplaybreaks

\DeclareMathOperator*{\diag}{diag}

\DeclareMathOperator*{\var}{var}
\DeclareMathOperator*{\tr}{tr}

\DeclareMathOperator*{\poisson}{Poisson}
\DeclareMathOperator*{\median}{median}
\DeclareMathOperator*{\MAD}{MAD}
\newcommand*{\doi}[1]{\href{http://dx.doi.org/#1}{doi: #1}}
\newcommand{\bs}[1]{\boldsymbol{#1}}
\DeclareFloatingEnvironment[name={Supplementary Figure}]{suppfigure}

\begin{document}
\singlespacing
\title{Unifying and Generalizing Methods for Removing Unwanted Variation Based on Negative Controls}
\author{David Gerard$^1$ and Matthew Stephens$^{1,2}$ \\
  Departments of Human Genetics$^1$ and Statistics$^2$, \\
  University of Chicago, Chicago, IL, 60637, USA}
\date{May 23, 2017}
\maketitle

\begin{abstract}
  Unwanted variation, including hidden confounding, is a well-known
  problem in many fields, particularly large-scale gene expression
  studies. Recent proposals to use control genes --- genes assumed to
  be unassociated with the covariates of interest --- have led to new
  methods to deal with this problem. Going by the moniker
  \textbf{R}emoving \textbf{U}nwanted \textbf{V}ariation (RUV), there
  are many versions --- RUV1, RUV2, RUV4, RUVinv, RUVrinv, RUVfun.  In
  this paper, we introduce a general framework, RUV*, that both unites
  and generalizes these approaches. This unifying framework helps
  clarify connections between existing methods. In particular we
  provide conditions under which RUV2 and RUV4 are equivalent. The
  RUV* framework also preserves an advantage of RUV approaches ---
  their modularity --- which facilitates the development of novel
  methods based on existing matrix imputation algorithms. We
  illustrate this by implementing RUVB, a version of RUV* based on
  Bayesian factor analysis.  In realistic simulations based on real
  data we found that RUVB is competitive with existing methods in
  terms of both power and calibration, although we also highlight the
  challenges of providing consistently reliable calibration among data
  sets.
\end{abstract}

\renewcommand{\thefootnote}{\fnsymbol{footnote}}
\renewcommand{\thefootnote}{\arabic{footnote}}

\section{Introduction}
\label{section:introduction}
Many experiments and observational studies in genetics are overwhelmed
with unwanted sources of variation. Examples include:\ processing date
\citep{akey2007design}, the lab that collected a sample
\citep{irizarry2005multiple}, the batch in which a sample was
processed \citep{leek2010tackling}, and subject attributes such as
environmental factors \citep{gibson2008environmental} and ancestry
\citep{price2006principal}. These factors, if ignored, can result in
disastrously wrong conclusions \citep{gilad2015reanalysis}. They can
induce dependencies between samples, and inflate test statistics,
making it difficult to control false discovery rates
\citep{efron2004large, efron2008microarrays,efron2010correlated}.

Many of the sources of variation mentioned above are likely to be
observed, in which case standard methods exist to control for them
\citep{johnson2007adjusting}. However, every study likely also
contains unobserved sources of unwanted variation, and these can cause
equally profound problems \citep{leek2007capturing} --- even in the
ideal case of a randomized experiment. To illustrate this we took $20$
samples from an RNA-seq dataset \citep{gtex2015} and randomly assigned
them into two groups of 10 samples. Since group assignment is entirely
independent of the expression levels of each gene, the group labels
are theoretically unassociated with all genes and any observed
``signal'' must be artefactual. Figure \ref{figure:all.null} shows
histograms of the $p$-values from two-sample $t$-tests for three
different randomizations. In each case the distribution of the
$p$-values differs greatly from the theoretical uniform
distribution. Thus, even in this ideal scenario where group labels
were randomly assigned, problems can arise.  One way to understand
this is to note that the same randomization is being applied to all
genes. Consequently, if many genes are affected by an unobserved
factor, and this factor happens by chance to be correlated with the
randomization, then the $p$-value distributions will be non-uniform.
In this sense the problems here can be viewed as being due to
correlation among the $p$ values; see \citet{efron2010correlated} for
extensive discussion.  (The issue of whether the problems in any given
study are caused by correlation, confounding, or something different
is both interesting and subtle; see discussion in
\citet{efron2010correlated, schwartzmann2010comment} for example. For this
reason we adopt the ``unwanted variation'' terminology from
\citet{gagnon2012using}, rather than alternative terminologies such as
``hidden confounding''.)

\begin{figure}
\begin{center}
\includegraphics{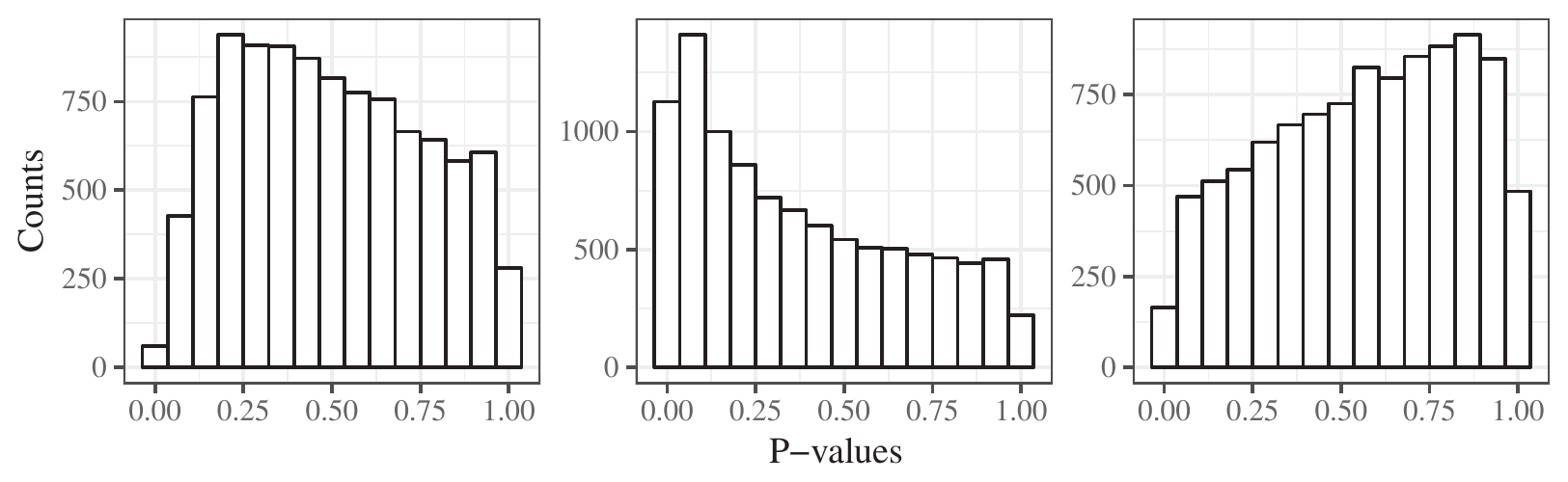}
\caption{Histograms of $p$-values from two-sample $t$-tests when group labels are
  randomly assigned to samples. Each panel is from a different random
  seed. The $p$-value distributions all clearly deviate from uniform.}
\label{figure:all.null}
\end{center}
\end{figure}

In recent years many methods have been introduced to try to solve
problems due to unwanted variation. Perhaps the simplest approach is
to estimate sources of unwanted variation using principal components
analysis \citep{price2006principal}, and then to control for these
factors by using them as covariates in subsequent analyses. Indeed, in
genome-wide association studies this simple method is widely used.
However, in gene expression studies it suffers from the problem that
the principal components will typically also contain the signal of
interest, so controlling for them risks removing that signal.  To
address this \citet{leek2007capturing,leek2008general} introduced
\textbf{S}urrogate \textbf{V}ariable \textbf{A}nalysis (SVA), which
uses an iterative algorithm to attempt to estimate latent factors that
do not include the signal of interest (see also
\citet{lucas2006sparse} for related work).  To account for unwanted
variation, SVA assumes a factor-augmented regression model (see
Section \ref{subsection:rotate}), from which there is an old
literature \citep[and others]{fisher1923studies,
  cochran1943comparison, williams1952interpretation, tukey1962future,
  gollob1968statistical, mandel1969partitioning, mandel1971new,
  efron1972empirical, freeman1973statistical, gabriel1978least}. Since
SVA, a large number of different approaches have emerged along similar
lines, including \citet{behzadi2007component, carvalho2008high,
  kang2008accurate, kang2008efficient, stegle2008accounting,
  friguet2009factor, kang2010variance, listgarten2010correction,
  stegle2010bayesian, wu2010subset, gagnon2012using, fusi2012joint,
  stegle2012using, sun2012multiple, gagnon2013removing,
  mostafavi2013normalizing, perry2013degrees, yang2013accounting,
  wang2015confounder, chen2016normalization}, among others.

As noted above, a key difficulty in adjusting for unwanted variation
in expression studies is distinguishing between the effect of a
treatment and the effect of factors that are correlated with a
treatment. Available methods deal with this problem in different ways.
In this paper we focus on the subset of methods that use ``negative
controls'' to help achieve this goal. In the context of a gene
expression study, a negative control is a gene whose expression is
assumed \emph{a priori} to be unassociated with all covariates (and
treatments) of interest.  Under this assumption, negative controls can
be used to separate sources of unwanted variation from the treatment
effects.  The idea of using negative controls in this way appears in
\citet{lucas2006sparse}, and has been recently popularized by
\citet{gagnon2012using,gagnon2013removing} in a series of methods and
software going by the moniker \textbf{R}emoving \textbf{U}nwanted
\textbf{V}ariation (RUV).  There are several different methods,
including RUV2 (for RUV 2-step), RUV4, RUVinv (a special case of
RUV4), RUVrinv, RUVfun, and RUV1.

Understanding the relative merits and properties of the different RUV
methods, which are all aimed at solving essentially the same problem,
is a non-trivial task. The main contribution of this paper is to
outline a general framework, RUV*, that encompasses all versions of
RUV (Section \ref{section:ruv5.full}). RUV* represents the problem as
a general matrix imputation procedure, both providing a unifying
conceptual framework, and opening up new approaches based on the large
literature in matrix imputation. Our RUV* framework also provides a
simple and modular way to account for uncertainty in the estimated
sources of unwanted variation, which is an issue ignored by most
methods.  On the way to this general framework we make detailed
connections between RUV2 and RUV4, exploiting the formulation in
\citet{wang2015confounder}.

The structure of the paper is as follows.  We first review the
formulation of \citet{wang2015confounder} (Section
\ref{subsection:rotate}) in the context of RUV4 (Section
\ref{subsection:variance}). We then extend this to include RUV2
(Section \ref{section:ruv2}), and develop necessary and sufficient
conditions for a procedure to be a version of both RUV2 and RUV4
(Section \ref{section:ruv3.full}). We call the resulting procedure
RUV3. We then develop the general RUV* framework (Section
\ref{section:ruv5.full}), and illustrate it by implementing an
approach based on Bayesian factor analysis that we call RUVB. In
Section \ref{section:variance.estimation} we briefly discuss the issue
of variance estimation, which turns out to be an important issue that
can greatly affect empirical performance.  In Section
\ref{section:evaluate} we use realistic simulations, based on real
data, to compare many of the available methods that use negative
controls. We show that RUVB has competitive power and calibration
compared with other methods, particularly when there are few control
genes. We finish with a discussion in Section
\ref{section:discussion}.

On notation:\ throughout we denote matrices using bold capital letters
($\bs{A}$), except for $\bs{\alpha}$ and $\bs{\beta}$, which are also
matrices. Bold lowercase letters are vectors ($\bs{a}$), and non-bold
lowercase letters are scalars ($a$). Where there is no chance for
confusion, we use non-bold lowercase to denote scalar elements of
vectors or matrices. For example, $a_{ij}$ is the $(i,j)$th element of
$\bs{A}$ and $a_i$ is the $i$th element of $\bs{a}$. The notation
$\bs{A}_{n \times m}$ denotes that the matrix $\bs{A}$ is an $n$ by
$m$ matrix. The matrix transpose is denoted $\bs{A}^{\intercal}$ and
the matrix inverse is denoted $\bs{A}^{-1}$. Sets are generally
denoted with calligraphic letters ($\mathcal{A}$), and the complement
of a set is denoted with a bar ($\bar{\mathcal{A}}$). Finally,
$p(\bs{a}|\bs{b})$ denotes a probability density function of $\bs{a}$
conditional on $\bs{b}$.

\section{RUV4}
\label{section:methods}

\subsection{Review of the Two-step Rotation Method}
\label{subsection:rotate}

Most existing approaches to this problem \citep{leek2007capturing,
  leek2008general, gagnon2012using, sun2012multiple,
  gagnon2013removing, wang2015confounder} use a low-rank matrix to
capture unwanted variation.  Specifically, they assume:
\begin{align}
  \label{equation:full.model}
  \bs{Y}_{n\times p} = \bs{X}_{n \times
    k}\bs{\beta}_{k \times p} + \bs{Z}_{n \times
    q}\bs{\alpha}_{q \times p} + \bs{E}_{n\times
    p},
\end{align}
where, in the context of a gene-expression study, $y_{ij}$ is the
normalized expression level of the $j$th gene on the $i$th sample,
$\bs{X}$ contains the observed covariates, $\bs{\beta}$ contains the
coefficients of $\bs{X}$, $\bs{Z}$ is a matrix of unobserved factors
(sources of unwanted variation), $\bs{\alpha}$ contains the
coefficients of $\bs{Z}$, and $\bs{E}$ contains independent (Gaussian)
errors with means 0 and column-specific variances
$\var(e_{ij}) = \sigma_j^2$. In this model, the only known quantities
are $\bs{Y}$ and $\bs{X}$.

To fit \eqref{equation:full.model}, it is common to apply a two-step
approach
(e.g.\ \citet{gagnon2013removing,sun2012multiple,wang2015confounder}).
The first step regresses out $\bs{X}$ and then, using the residuals of
this regression, estimates $\bs{\alpha}$ and the $\sigma_j$'s.  The
second step then assumes that $\bs{\alpha}$ and the $\sigma_j$'s are
known and estimates $\bs{\beta}$ and $\bs{Z}$.
\citet{wang2015confounder} helpfully frame this two-step approach as a
rotation followed by estimation in two independent models. We now
review this approach.

First, we let $\bs{X} = \bs{Q}\bs{R}$ denote the QR decomposition of
$\bs{X}$, where $\bs{Q} \in \mathbb{R}^{n \times n}$ is an orthogonal
matrix
($\bs{Q}^{\intercal}\bs{Q} = \bs{Q}\bs{Q}^{\intercal} = \bs{I}_n$) and
$\bs{R}_{n \times k} = {{\bs{R}_1}\choose{0}}$, where
$\bs{R}_1 \in \mathbb{R}^{k \times k}$ is an upper-triangular
matrix. Multiplying \eqref{equation:full.model} on the left by
$\bs{Q}^{\intercal}$ yields
\begin{align}
  \label{equation:rotated}
  \bs{Q}^{\intercal}\bs{Y} =
  \bs{R}\bs{\beta} +
  \bs{Q}^{\intercal}\bs{Z}\bs{\alpha} +
  \bs{Q}^{\intercal}\bs{E}.
\end{align}

Suppose that $k = k_1 + k_2$, where the first $k_1$ covariates of
$\bs{X}$ are not of direct interest, but are included because of
various modeling decisions (e.g.\ an intercept term, or covariates that
need to be controlled for). The last $k_2$ columns of $\bs{X}$ are the
variables of interest whose putative associations with $\bs{Y}$ the
researcher wishes to test. Let
$\bs{Y}_1 \in \mathbb{R}^{k_1 \times p}$ be the first $k_1$ rows of
$\bs{Q}^{\intercal}\bs{Y}$, $\bs{Y}_2 \in \mathbb{R}^{k_2 \times p}$
be the next $k_2$ rows of $\bs{Q}^{\intercal}\bs{Y}$, and
$\bs{Y}_3 \in \mathbb{R}^{(n-k) \times p}$ be the last $n-k$ rows of
$\bs{Q}^{\intercal}\bs{Y}$. Conformably partition
$\bs{Q}^{\intercal}\bs{Z}$ into $\bs{Z}_1$, $\bs{Z}_2$, and
$\bs{Z}_3$, and $\bs{Q}^{\intercal}\bs{E}$ into $\bs{E}_1$,
$\bs{E}_2$, and $\bs{E}_3$. Let
\begin{align}
  \bs{R}_1 =
  \left(
    \begin{array}{cc}
      \bs{R}_{11} & \bs{R}_{12}\\
      0 & \bs{R}_{22}
    \end{array}
  \right).
\end{align}
Finally, partition
$\bs{\beta} = {{\bs{\beta}_1}\choose{\bs{\beta}_2}}$ so that
$\bs{\beta}_1 \in \mathbb{R}^{k_1 \times p}$ contains the coefficients
for the first $k_1$ covariates and
$\bs{\beta}_2 \in \mathbb{R}^{k_2\times p}$ contains the coefficients
for the last $k_2$ covariates. Then \eqref{equation:rotated} may be
written as three models
\begin{align}
  \label{equation:Y1.model}
  \bs{Y}_1 &= \bs{R}_{11}\bs{\beta}_1 + \bs{R}_{12}\bs{\beta}_2 + \bs{Z}_1\bs{\alpha} + \bs{E}_1,\\
  \label{equation:Y2.model}
  \bs{Y}_2 &= \phantom{\bs{R}_{11}\bs{\beta}_1 +\ } \bs{R}_{22}\bs{\beta}_2 + \bs{Z}_2\bs{\alpha} + \bs{E}_2,\\
  \label{equation:Y3.model}
  \bs{Y}_3 &=
  \phantom{\bs{R}_{11}\bs{\beta}_1 +
    \bs{R}_{12}\bs{\beta}_2 +\ }
  \bs{Z}_3\bs{\alpha} + \bs{E}_3.
\end{align}

Importantly, the error terms in \eqref{equation:Y1.model},
\eqref{equation:Y2.model}, and \eqref{equation:Y3.model} are mutually
independent. This follows from the easily-proved fact that $\bs{E}$ is
equal in distribution to $\bs{Q}^{\intercal}\bs{E}$.  The two-step
estimation procedure mentioned above becomes:\ first, estimate
$\bs{\alpha}$ and the $\sigma_j$'s using \eqref{equation:Y3.model};
second, estimate $\bs{\beta}_2$ and $\bs{Z}_2$ given $\bs{\alpha}$ and
the $\sigma_j$'s using \eqref{equation:Y2.model}. Equation
\eqref{equation:Y1.model} contains the nuisance parameters
$\bs{\beta}_1$ and is ignored.

\subsection{Review of RUV4}
\label{subsection:variance}

One approach to distinguishing between unwanted variation and effects
of interest is to use ``control genes'' \citep{lucas2006sparse,
  gagnon2012using}.  A control gene is a gene that is assumed \emph{a
  priori} to be unassociated with the covariate(s) of interest. More
formally, the set of control genes,
$\mathcal{C} \subseteq \{1, \ldots, p\}$, has the property that
\begin{align}
  \beta_{ij} = 0 \text{ for all } i = k_1 + 1, \ldots, k, \text{ and } j \in \mathcal{C},
\end{align}
and is a subset of the truly null genes.  Examples of control genes
used in practice are spike-in controls \citep{jiang2011synthetic} used
to adjust for technical factors (such as sample batch) and
housekeeping genes \citep{eisenberg2013human} used to adjust for both
technical and biological factors (such as subject ancestry).

RUV4 \citep{gagnon2013removing} uses control genes to estimate
$\bs{\beta}_2$ in the presence of unwanted variation. Let
$\bs{Y}_{2\mathcal{C}} \in \mathbb{R}^{k_2 \times m}$ denote the
submatrix of $\bs{Y}_{2}$ with columns that correspond to the $m$
control genes. Similarly subset the relevant columns to obtain
$\bs{\beta}_{2\mathcal{C}} \in \mathbb{R}^{k_2\times m}$,
$\bs{\alpha}_{\mathcal{C}}\in\mathbb{R}^{q\times m}$, and
$\bs{E}_{2\mathcal{C}} \in \mathbb{R}^{k_2 \times m}$.  The steps for
RUV4, including a variation from \citet{wang2015confounder}, are
presented in Procedure \ref{algorithm:ruv4}. (For simplicity we focus
on point estimates of effects here, deferring assessment of standard
errors to Section \ref{section:variance.estimation}.)

The key idea in Procedure \ref{algorithm:ruv4} is that for the
control genes model \eqref{equation:Y2.model} becomes
\begin{align}
  \nonumber\bs{Y}_{2\mathcal{C}} &= \bs{R}_{22}\bs{\beta}_{2\mathcal{C}} + \bs{Z}_2\hat{\bs{\alpha}}_{\mathcal{C}} + \bs{E}_{2\mathcal{C}},\\
  &= \phantom{\bs{R}_{22}\bs{\beta}_{2\mathcal{C}}
    +\ } \bs{Z}_2\hat{\bs{\alpha}}_{\mathcal{C}} +
  \bs{E}_{2\mathcal{C}}, \label{equation:control.model}\\
  e_{2\mathcal{C}ij} &\overset{ind}{\sim} N(0,
  \hat{\sigma}_j^2) \label{equation:control.model.var}.
\end{align}
The equality in \eqref{equation:control.model} follows from the
property of control genes that
$\bs{\beta}_{2\mathcal{C}} = \bs{0}$. Step 2 of
Procedure \ref{algorithm:ruv4} uses \eqref{equation:control.model} to
estimate $\bs{Z}_2$.

Step 1 of Procedure \ref{algorithm:ruv4} requires a factor analysis of $\bs{Y}_3$. We formally define a
factor analysis as follows.
\begin{definition}
  \label{def:fa}
  A \emph{factor analysis}, $\mathcal{F}$, of rank $q \leq \min(n, p)$
  on $\bs{Y} \in \mathbb{R}^{n \times p}$ is a set of three
  functions $\mathcal{F} = \{\hat{\bs{\Sigma}}(\bs{Y}),
  \hat{\bs{Z}}(\bs{Y}),
  \hat{\bs{\alpha}}(\bs{Y})\}$ such that
  $\hat{\bs{\Sigma}}(\bs{Y}) \in \mathbb{R}^{p \times
    p}$ is diagonal with positive diagonal entries,
  $\hat{\bs{Z}}(\bs{Y}) \in \mathbb{R}^{n \times q}$ has
  rank $q$, and $\hat{\bs{\alpha}}(\bs{Y}) \in
  \mathbb{R}^{q \times p}$ has rank $q$.
\end{definition}
RUV4 allows the analyst to use any factor analysis they desire. Thus,
RUV4 is not a single method, but a collection of methods indexed by
the factor analysis used. When we want to be explicit about this
indexing, we will write RUV4$(\mathcal{F})$.

\begin{algorithm}
  \floatname{algorithm}{Procedure}
  \caption{RUV4}
  \label{algorithm:ruv4}
  \begin{algorithmic}[1]
    \STATE Estimate $\bs{\alpha}$ and $\bs{\Sigma}$ using a
    factor analysis (Definition \ref{def:fa}) on $\bs{Y}_3$ in
    \eqref{equation:Y3.model}. Call these estimates
    $\hat{\bs{\alpha}}$ and $\hat{\bs{\Sigma}}$.
    \STATE Estimate $\bs{Z}_2$ using control genes (equation \eqref{equation:control.model}). Let
    $\hat{\bs{\Sigma}}_{\mathcal{C}} =
    \diag(\hat{\sigma}_{j_1}^2,\ldots, \hat{\sigma}_{j_m}^2)$ for $j_i
    \in \mathcal{C}$ for all $i = 1,\ldots, m$.

    RUV4 in \citet{gagnon2013removing} estimates $\bs{Z}_2$ by ordinary least squares (OLS)
    \begin{align}
      \label{equation:ruv4.gag.z2.est}
      \hat{\bs{Z}}_2 = \bs{Y}_{2\mathcal{C}}
      \hat{\bs{\alpha}}_{\mathcal{C}}^{\intercal}(\hat{\bs{\alpha}}_{\mathcal{C}}\hat{\bs{\alpha}}_{\mathcal{C}}^{\intercal})^{-1}.
    \end{align}

Alternatively, \citet{wang2015confounder} implement a variation on RUV4 (which we call CATE, and is implemented in the R package {\tt cate}) that
estimates $\bs{Z}_2$ by generalized least squares (GLS)
    \begin{align}
      \label{equation:cate.z2.est}
      \hat{\bs{Z}}_2 = \bs{Y}_{2\mathcal{C}}
      \hat{\bs{\Sigma}}_{\mathcal{C}}^{-1}
      \hat{\bs{\alpha}}_{\mathcal{C}}^{\intercal}(\hat{\bs{\alpha}}_{\mathcal{C}}\hat{\bs{\Sigma}}_{\mathcal{C}}^{-1}\hat{\bs{\alpha}}_{\mathcal{C}}^{\intercal})^{-1}.
    \end{align}

    \STATE Estimate $\bs{\beta}_2$ using
    \eqref{equation:Y2.model} by
    \begin{align}
      \hat{\bs{\beta}}_2 = \bs{R}_{22}^{-1}(\bs{Y}_2 -
      \hat{\bs{Z}}_2\hat{\bs{\alpha}}).
    \end{align}
  \end{algorithmic}
\end{algorithm}

\section{RUV2}
\label{section:ruv2}

RUV2 is a method for removing unwanted variation presented in
\citet{gagnon2012using} (predating RUV4 in
\citet{gagnon2013removing}). One contribution of our paper is to
present RUV2 within the same rotation framework as RUV4 (Section
\ref{subsection:rotate}).

Procedure \ref{algorithm:ruv2.gag} summarizes RUV2 as presented in
\citet{gagnon2012using}.  The method is simple:\ first estimate the
factors causing unwanted variation from the control genes, and then
include these factors as covariates in the regression models for the
non-control genes. However, this procedure does not deal with nuisance
parameters. To deal with these, \citet{gagnon2013removing} introduce
an extension, also called RUV2. This extension first rotates $\bs{Y}$
and $\bs{X}$ onto the orthogonal complement of the columns of $\bs{X}$
corresponding to the nuisance parameters \citep[equation (64)
in][]{gagnon2013removing} and then applies Procedure
\ref{algorithm:ruv2.gag}.

Like RUV4, RUV2 is actually a class of methods indexed by the factor
analysis used (Definition \ref{def:fa}). We denote the class of RUV2
methods presented in \citet{gagnon2012using} by
$\text{RUV2}_{old}(\mathcal{F})$.

In Procedure \ref{algorithm:ruv2.new} we present a class of methods,
denoted $\text{RUV2}_{new}(\mathcal{F})$, that we will prove to be
equivalent to $\text{RUV2}_{old}$.

\begin{algorithm}
  \floatname{algorithm}{Procedure}
  \caption{RUV2 (without nuisance covariates) as presented in \citet{gagnon2012using}}
  \label{algorithm:ruv2.gag}
  \begin{algorithmic}[1]
    \STATE From \eqref{equation:full.model}, estimate $\bs{Z}$ by factor analysis on $\bs{Y}_{\mathcal{C}}$. Call this estimate $\hat{\bs{Z}}$.
    \STATE Estimate $\bs{\beta}$ by regressing $\bs{Y}$ on $(\bs{X}, \hat{\bs{Z}})$. That is
    \begin{align}
      \hat{\bs{\beta}} = (\bs{X}^{\intercal}\bs{S}\bs{X})^{-1}\bs{X}^{\intercal}\bs{S}\bs{Y},
    \end{align}
    where $\bs{S} := \bs{I}_n - \hat{\bs{Z}}(\hat{\bs{Z}}^{\intercal}\hat{\bs{Z}})^{-1}\hat{\bs{Z}}^{\intercal}$.
  \end{algorithmic}
\end{algorithm}

\begin{algorithm}
  \floatname{algorithm}{Procedure}
  \caption{RUV2 in rotated model framework of Section \ref{subsection:rotate}}
  \label{algorithm:ruv2.new}
  \begin{algorithmic}[1]
    \STATE Estimate $\bs{Z}_2$ and $\bs{Z}_3$ by factor
    analysis on ${\bs{Y}_{2\mathcal{C}} \choose
      \bs{Y}_{3\mathcal{C}}}$. Call these estimates
    $\hat{\bs{Z}}_2$ and $\hat{\bs{Z}}_3$.
    \STATE Estimate $\bs{\alpha}$ and $\bs{\Sigma}$ by
    regressing $\bs{Y}_3$ on $\hat{\bs{Z}}_3$. That is
    \begin{align}
      \label{equation:alphahat.ruv2.rotate}
      \hat{\bs{\alpha}} &= (\hat{\bs{Z}}_3^{\intercal}\hat{\bs{Z}}_3^{-1})\hat{\bs{Z}}_3^{\intercal}\bs{Y}_3 \text{ and}\\
      \hat{\bs{\Sigma}} &= \diag[(\bs{Y}_3 - \hat{\bs{Z}}_3\hat{\bs{\alpha}})^{\intercal}(\bs{Y}_3 - \hat{\bs{Z}}_3\hat{\bs{\alpha}})] / (n - k - q).
    \end{align}
    \STATE Estimate $\bs{\beta}_2$ with
    \begin{align}
      \label{equation:beta2.ruv2.new}
      \hat{\bs{\beta}}_2 = \bs{R}_{22}^{-1}(\bs{Y}_{2} - \hat{\bs{Z}}_2\hat{\bs{\alpha}}).
    \end{align}
  \end{algorithmic}
\end{algorithm}

\begin{theorem}
\label{theorem:ruv2.equal}
  For a given orthogonal matrix $\bs{Q} \in \mathbb{R}^{n
    \times n}$ and an arbitrary non-singular matrix $\bs{A}(\bs{Y})$
  that (possibly) depends on $\bs{Y}$, suppose
  \begin{align}
    \label{equation:fa.gag}\mathcal{F}_1(\bs{Y}) &:= \{\hat{\bs{\Sigma}}(\bs{Y}), \hat{\bs{Z}}(\bs{Y}), \hat{\bs{\alpha}}(\bs{Y})\}, \text{ and}\\
    \label{equation:fa.new}\mathcal{F}_2(\bs{Y}) &:= \{\hat{\bs{\Sigma}}(\bs{Q}^{\intercal}\bs{Y}), \bs{Q}\hat{\bs{Z}}(\bs{Q}^{\intercal}\bs{Y})\bs{A}(\bs{Y}), \bs{A}^{-1}(\bs{Y})\hat{\bs{\alpha}}(\bs{Q}^{\intercal}\bs{Y})\}.
  \end{align}
  Then
\begin{align}
\text{RUV2}_{old}(\mathcal{F}_2) = \text{RUV2}_{new}(\mathcal{F}_1).
\end{align}
That is, Procedure \ref{algorithm:ruv2.gag} using factor analysis
\eqref{equation:fa.gag} is equivalent to Procedure
\ref{algorithm:ruv2.new} using factor analysis
\eqref{equation:fa.new}.
\end{theorem}
\begin{proof}
  For simplicity, we first assume that there are no nuisance
  covariates. Then
  ${\bs{Y}_{2\mathcal{C}} \choose \bs{Y}_{3\mathcal{C}}} =
  \bs{Q}^{\intercal}\bs{Y}$,
  where $\bs{Q}$ is the orthogonal matrix from the QR decomposition of
  $\bs{X}$ (Section \ref{subsection:rotate}). Thus,
  ${\hat{\bs{Z}}_2 \choose \hat{\bs{Z}}_3}$ from Procedure
  \ref{algorithm:ruv2.new} results from applying $\mathcal{F}_1$ on
  $\bs{Q}^{\intercal}\bs{Y}$ while $\hat{\bs{Z}}$ from Procedure
  \ref{algorithm:ruv2.gag} results from applying $\mathcal{F}_2$ on
  $\bs{Y}$. From the definitions of \eqref{equation:fa.gag} and
  \eqref{equation:fa.new}, we thus have that $\hat{\bs{Z}}$ from step
  1 of Procedure \ref{algorithm:ruv2.gag} is in the same column space
  as $\bs{Q}{\hat{\bs{Z}}_2 \choose \hat{\bs{Z}}_3}$ from step 1 of
  Procedure \ref{algorithm:ruv2.new}. $\hat{\bs{\beta}}_2$ from
  \eqref{equation:beta2.ruv2.new} contains the partial regression
  coefficients of $\bs{X}$ when including
  $\bs{Q}{\hat{\bs{Z}}_2 \choose \hat{\bs{Z}}_3}$ as nuisance
  covariates (to show this, just calculate the MLE's of $\bs{\beta}_2$
  and $\bs{\alpha}$ using \eqref{equation:Y2.model} and
  \eqref{equation:Y3.model}). The estimates of $\bs{\beta}_2$ in step
  2 of Procedure \ref{algorithm:ruv2.gag} are also partial regression
  coefficients of $\bs{X}$ when including $\hat{\bs{Z}}$ as nuisance
  covariates. Since the partial regression coefficients in Procedure
  \ref{algorithm:ruv2.gag} are only a function of $\hat{\bs{Z}}$
  through its column space, and the partial regression coefficients in
  Procedure \ref{algorithm:ruv2.new} are only a function of
  $\bs{Q}{\hat{\bs{Z}}_2 \choose \hat{\bs{Z}}_3}$ through its column
  space, and these column spaces are the same, we have completed the
  proof for the case of no nuisance parameters.

  To deal with nuisance parameters, \citet{gagnon2013removing} rotate
  $\bs{X}$ and $\bs{Y}$ onto the orthogonal complement of the columns
  of $\bs{X}$ corresponding to the nuisance parameters prior to
  applying Procedure \ref{algorithm:ruv2.gag}. If we partition
  $\bs{Q} = (\bs{Q}_1, \bs{Q}_2, \bs{Q}_3)$ and
  $\bs{X} = (\bs{X}_1, \bs{X}_2)$, this is equivalent to using the
  model
\begin{align}
\label{equation:reduced.model}
\bs{W}
\left(
\begin{array}{c}
\bs{Q}_2^{\intercal}\\
\bs{Q}_3^{\intercal}
\end{array}
\right)
\bs{Y}
=
\bs{W}
\left(
\begin{array}{c}
\bs{Q}_2^{\intercal}\\
\bs{Q}_3^{\intercal}
\end{array}
\right)
\bs{X}_2\bs{\beta}_2 +
\bs{W}
\left(
\begin{array}{c}
\bs{Q}_2^{\intercal}\\
\bs{Q}_3^{\intercal}
\end{array}
\right)
\bs{Z}\bs{\alpha} +
\bs{W}
\left(
\begin{array}{c}
\bs{Q}_2^{\intercal}\\
\bs{Q}_3^{\intercal}
\end{array}
\right)
\bs{E}
\end{align}
where $\bs{W}$ is some arbitrary (but known) $n - k_1$ by $n - k_1$
orthogonal matrix. Or, using the QR decomposition of $\bs{X}$,
\eqref{equation:reduced.model}  is
equal to
\begin{align}
\bs{W}
\left(
\begin{array}{c}
\bs{Q}_2^{\intercal}\\
\bs{Q}_3^{\intercal}
\end{array}
\right)
\bs{Y}
=
\bs{W}
\left(
\begin{array}{c}
\bs{R}_{22}\\
\bs{0}
\end{array}
\right)
\bs{\beta}_2
+
\bs{W}
\left(
\begin{array}{c}
\bs{Q}_2^{\intercal}\\
\bs{Q}_3^{\intercal}
\end{array}
\right)
\bs{Z}\bs{\alpha} +
\bs{W}
\left(
\begin{array}{c}
\bs{Q}_2^{\intercal}\\
\bs{Q}_3^{\intercal}
\end{array}
\right)
\bs{E}.
\end{align}
Let
\begin{align}
\tilde{\bs{Y}} :=
\bs{W}
\left(
\begin{array}{c}
\bs{Q}_2^{\intercal}\\
\bs{Q}_3^{\intercal}
\end{array}
\right)
\bs{Y}, \
\tilde{\bs{X}} :=
\bs{W}
\left(
\begin{array}{c}
\bs{R}_{22}\\
\bs{0}
\end{array}
\right), \
\tilde{\bs{Z}} :=
\bs{W}
\left(
\begin{array}{c}
\bs{Q}_2^{\intercal}\\
\bs{Q}_3^{\intercal}
\end{array}
\right)
\bs{Z}, \text{ and }
\tilde{\bs{E}} :=
\bs{W}
\left(
\begin{array}{c}
\bs{Q}_2^{\intercal}\\
\bs{Q}_3^{\intercal}
\end{array}
\right)
\bs{E}.
\end{align}
Then, \eqref{equation:reduced.model} is equal to
\begin{align}
\label{equation:tilde.model}
\tilde{\bs{Y}} = \tilde{\bs{X}}\bs{\beta} + \tilde{\bs{Z}}\bs{\alpha} + \tilde{\bs{E}}.
\end{align}
We now just apply the arguments of the previous paragraph to
\eqref{equation:tilde.model}, where there are no nuisance parameters
and where the QR decomposition of $\tilde{\bs{X}}$ is just
$\tilde{\bs{X}} = \bs{W} {\bs{R}_{22}\choose
  \bs{0}}$. This completes the proof.
\end{proof}

The equivalence of $\text{RUV2}_\text{old}$ and
$\text{RUV2}_\text{new}$ in Theorem \ref{theorem:ruv2.equal} involves
using different factor analyses in each procedure. One can ask under
what conditions the two procedures would be equivalent if given the
{\it same} factor analysis.  We now prove that it suffices for the
factor analysis to be \emph{left orthogonally equivariant}.

\begin{definition}
\label{def:equivariant}
A factor analysis of rank $q$ on $\bs{Y} \in \mathbb{R}^{n \times p}$ is \emph{left orthogonally equivariant} if
\begin{align}
\{\hat{\bs{\Sigma}}(\bs{Q}^{\intercal}\bs{Y}), \hat{\bs{Z}}(\bs{Q}^{\intercal}\bs{Y})\bs{A}(\bs{Y}), \bs{A}(\bs{Y})^{-1}\hat{\bs{\alpha}}(\bs{Q}^{\intercal}\bs{Y})\} =
\{\hat{\bs{\Sigma}}(\bs{Y}), \bs{Q}^{\intercal}\hat{\bs{Z}}(\bs{Y}), \hat{\bs{\alpha}}(\bs{Y})\},
\end{align}
for all fixed orthogonal $\bs{Q} \in \mathbb{R} ^ {n \times
  n}$ and an arbitrary non-singular
$\bs{A}(\bs{Y}) \in \mathbb{R}^{q \times q}$ that
(possibly) depends on $\bs{Y}$.
\end{definition}

\begin{corollary}
Suppose $\mathcal{F}$ is a left orthogonally equivariant factor analysis. Then
\begin{align}
\text{RUV2}_{old}(\mathcal{F}) = \text{RUV2}_{new}(\mathcal{F}).
\end{align}
\end{corollary}
\begin{proof}
Suppose $\mathcal{F}_2$ from \eqref{equation:fa.new} is left orthogonally equivariant, then
\begin{align}
\mathcal{F}_2(\bs{Y}) &= \{\hat{\bs{\Sigma}}(\bs{Q}^{\intercal}\bs{Y}), \bs{Q}\hat{\bs{Z}}(\bs{Q}^{\intercal}\bs{Y})\bs{A}(\bs{Y}), \bs{A}^{-1}(\bs{Y})\hat{\bs{\alpha}}(\bs{Q}^{\intercal}\bs{Y})\}\\
&=  \{\hat{\bs{\Sigma}}(\bs{Y}), \bs{Q}\bs{Q}^{\intercal}\hat{\bs{Z}}(\bs{Y}), \hat{\bs{\alpha}}(\bs{Y})\}\\
&=  \{\hat{\bs{\Sigma}}(\bs{Y}), \hat{\bs{Z}}(\bs{Y}), \hat{\bs{\alpha}}(\bs{Y})\} = \mathcal{F}_1(\bs{Y}).
\end{align}
Let $\mathcal{F} := \mathcal{F}_1 = \mathcal{F}_2$. From the results
of Theorem \ref{theorem:ruv2.equal}, we have that
\begin{align}
\text{RUV2}_{old}(\mathcal{F}) = \text{RUV2}_{old}(\mathcal{F}_2) = \text{RUV2}_{new}(\mathcal{F}_1) = \text{RUV2}_{new}(\mathcal{F}).
\end{align}
\end{proof}

A well-known example of a factor analysis that is left orthogonally
equivariant is a truncated singular value decomposition (SVD). That
is, let $\bs{Y} = \bs{U}\bs{D}\bs{V}^{\intercal}$ be the SVD of
$\bs{Y}$. Let
$\bs{D}^{(q)} = \diag(d_{11},\ldots,d_{qq}, 0, \ldots, 0) \in
\mathbb{R}^{n \times n}$
be a diagonal matrix whose first $q$ diagonal elements are the same as
in $\bs{D}$ and the last $n - q$ diagonal elements are $0$. Then
\begin{align}
\label{equation:tsvd.z}\hat{\bs{Z}}(\bs{Y}) &= \bs{U}[\bs{D}^{(q)}]^{1 - \pi}\\
\label{equation:tsvd.alpha}\hat{\bs{\alpha}}(\bs{Y}) &= [\bs{D}^{(q)}]^{\pi}\bs{V}^{\intercal}\\
\label{equation:tsvd.sigma}\hat{\bs{\Sigma}}(\bs{Y}) &= \diag\left[\bs{V}(\bs{D} - \bs{D}^{(q)})^2\bs{V}^{\intercal}\right] / (n - q),
\end{align}
for any $\pi \in [0, 1]$. (Without loss of generality, for the
remainder of this paper, we let $\pi = 1$.) Notably, this factor
analysis is the only option in the R package \texttt{ruv}
\citep{bartsch2015ruv}.

For the rest of this paper, we use RUV2 to refer to Procedure
\ref{algorithm:ruv2.new} and not Procedure \ref{algorithm:ruv2.gag},
even if the factor analysis used is \emph{not} orthogonally
equivariant. (By Theorem \ref{theorem:ruv2.equal}, this corresponds to
Procedure \ref{algorithm:ruv2.gag} with some other factor analysis.)

\section{RUV3}
\label{section:ruv3.full}
\citet{gagnon2013removing} provide a lengthy heuristic discussion
comparing RUV2 with RUV4 (their section 3.4). However, they provide no
mathematical equivalencies.  In this section, we introduce RUV3, a
procedure that is both a version of RUV2 and a version RUV4. We show
that it is the only such procedure that is both RUV2 and RUV4. RUV3
uses a partitioned matrix imputation procedure from \citet{owen2016bi}
to estimate the hidden factors. The coefficients are then estimated as
in RUV2 and RUV4.

\subsection{The RUV3 Procedure}
The main goal in all methods is to estimate
$\bs{\beta}_{2\bar{\mathcal{C}}}$, the coefficients corresponding to
the non-control genes. This involves incorporating information from
four models
\begin{align}
\bs{Y}_{2\mathcal{C}} &= \bs{Z}_{2}\bs{\alpha}_{\mathcal{C}} + \bs{E}_{2\mathcal{C}}\\
\bs{Y}_{2\bar{\mathcal{C}}} &= \bs{R}_{22}\bs{\beta}_2 + \bs{Z}_{2}\bs{\alpha}_{\bar{\mathcal{C}}} + \bs{E}_{2\bar{\mathcal{C}}}\\
\bs{Y}_{3\mathcal{C}} &= \bs{Z}_{3}\bs{\alpha}_{\mathcal{C}} + \bs{E}_{3\mathcal{C}}\\
\bs{Y}_{3\bar{\mathcal{C}}} &= \bs{Z}_{3}\bs{\alpha}_{\bar{\mathcal{C}}} + \bs{E}_{3\bar{\mathcal{C}}}.
\end{align}
We can rearrange the rows and columns of ${\bs{Y}_{2} \choose
  \bs{Y}_{3}}$ to write these models in matrix form:
\begin{align}
\label{equation:matrix.form}
\left(
\begin{array}{cc}
\bs{Y}_{2\mathcal{C}} & \bs{Y}_{2\bar{\mathcal{C}}}\\
\bs{Y}_{3\mathcal{C}} & \bs{Y}_{3\bar{\mathcal{C}}}
\end{array}
\right)
=
\left(
\begin{array}{cc}
\bs{Z}_{2}\bs{\alpha}_{\mathcal{C}} + \bs{E}_{2\mathcal{C}} & \bs{R}_{22}\bs{\beta}_2 + \bs{Z}_{2}\bs{\alpha}_{\bar{\mathcal{C}}} + \bs{E}_{2\bar{\mathcal{C}}}\\
\bs{Z}_{3}\bs{\alpha}_{\mathcal{C}} + \bs{E}_{3\mathcal{C}} & \bs{Z}_{3}\bs{\alpha}_{\bar{\mathcal{C}}} + \bs{E}_{3\bar{\mathcal{C}}}
\end{array}
\right).
\end{align}

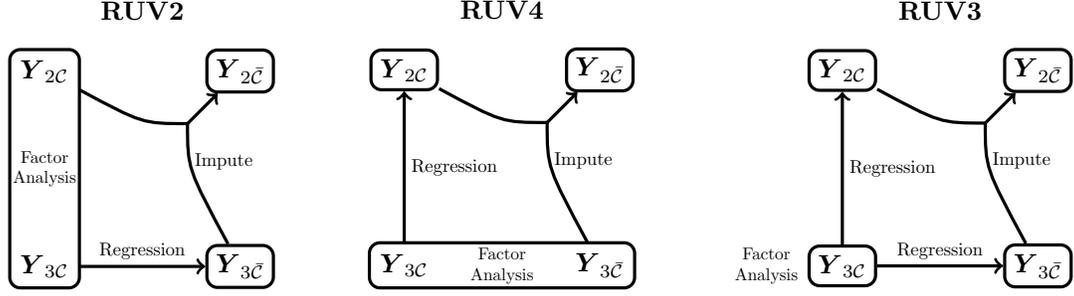
\begin{figure}
\begin{center}
\begin{tikzpicture}
\tikzstyle{every path}=[very thick];
\node at (0, 2.1){\textbf{RUV2}};
\path (-1.3, 1.3) node (y2c) {$\bs{Y}_{2\mathcal{C}}$}
(-1.3, -1.3) node (y3c) {$\bs{Y}_{3\mathcal{C}}$}
(1.3, -1.3) node[rounded corners, draw](y3cbar) {$\bs{Y}_{3\bar{\mathcal{C}}}$}
(1.3, 1.3) node[rounded corners, draw](y2cbar) {$\bs{Y}_{2\bar{\mathcal{C}}}$};
\coordinate (meetup) at (0.6, 0.6);
\path[->] (y3c) edge[above] node[scale = 0.7]{Regression} (y3cbar);
\draw (y3cbar) .. controls(0.6, 0) .. (meetup);
\draw (y2c) .. controls(0, 0.6) .. (meetup);
\draw[->] (meetup) -- (y2cbar);
\node[scale = 0.7, right](impute) at (0.6, 0.1){Impute};
\node[draw = black, fit = (y2c) (y3c), rounded corners, inner sep=0ex]{};
\node[scale = 0.65, align = center](fa) at (-1.3, 0){Factor\\Analysis};
\end{tikzpicture}
\hspace{1cm}
\begin{tikzpicture}
\tikzstyle{every path}=[very thick];
\node at (0, 2.1){\textbf{RUV4}};
\path (-1.3, 1.3) node[rounded corners, draw](y2c) {$\bs{Y}_{2\mathcal{C}}$}
(-1.3, -1.3) node (y3c) {$\bs{Y}_{3\mathcal{C}}$}
(1.3, -1.3) node (y3cbar) {$\bs{Y}_{3\bar{\mathcal{C}}}$}
(1.3, 1.3) node[rounded corners, draw](y2cbar) {$\bs{Y}_{2\bar{\mathcal{C}}}$};
\coordinate (meetup) at (0.6, 0.6);
\path[->] (y3c) edge[right] node[scale = 0.7]{Regression} (y2c);
\draw (y3cbar) .. controls(0.6, 0) .. (meetup);
\draw (y2c) .. controls(0, 0.6) .. (meetup);
\draw[->] (meetup) -- (y2cbar);
\node[scale = 0.7, right](impute) at (0.6, 0.1){Impute};
\node[draw = black, fit = (y3c) (y3cbar), rounded corners, inner sep=0ex]{};
\node[scale = 0.65, align = center](fa) at (0, -1.3){Factor\\Analysis};
\end{tikzpicture}
\hspace{1cm}
\begin{tikzpicture}
\tikzstyle{every path}=[very thick];
\node at (0, 2.1){\textbf{RUV3}};
\path (-1.3, 1.3) node[rounded corners, draw] (y2c) {$\bs{Y}_{2\mathcal{C}}$}
(-1.3, -1.3) node[rounded corners, draw] (y3c) {$\bs{Y}_{3\mathcal{C}}$}
(1.3, -1.3) node[rounded corners, draw](y3cbar) {$\bs{Y}_{3\bar{\mathcal{C}}}$}
(1.3, 1.3) node[rounded corners, draw](y2cbar) {$\bs{Y}_{2\bar{\mathcal{C}}}$};
\coordinate (meetup) at (0.6, 0.6);
\path[->] (y3c) edge[above] node[scale = 0.7]{Regression} (y3cbar);
\path[->] (y3c) edge[right] node[scale = 0.7]{Regression} (y2c);
\draw (y3cbar) .. controls(0.6, 0) .. (meetup);
\draw (y2c) .. controls(0, 0.6) .. (meetup);
\draw[->] (meetup) -- (y2cbar);
\node[scale = 0.7, right](impute) at (0.6, 0.1){Impute};
\node[scale = 0.65, align = center, left](fa) at (-1.8, -1.3){Factor\\Analysis};
\end{tikzpicture}
\caption{Pictorial representation of the differences between RUV2, RUV4, and RUV3.}
\label{fig:ruv.rep}
\end{center}
\end{figure}

The major difference between RUV2 and RUV4 is how the estimation
procedures interact in \eqref{equation:matrix.form}; see Figure
\ref{fig:ruv.rep} for a pictorial
representation. RUV2 performs factor analysis on
${\bs{Y}_{2\mathcal{C}} \choose \bs{Y}_{3\mathcal{C}}}$, then
regresses $\bs{Y}_{3\bar{\mathcal{C}}}$ on the estimated factor loadings. RUV4 performs factor
analysis on $(\bs{Y}_{3\mathcal{C}}, \bs{Y}_{3\bar{\mathcal{C}}})$,
then regresses $\bs{Y}_{2\mathcal{C}}$ on the estimated factors. The main goal in both, however, is to
estimate $\bs{Z}_{2}\bs{\alpha}_{\bar{\mathcal{C}}}$ given
$\bs{Y}_{2\mathcal{C}}$, $\bs{Y}_{3\mathcal{C}}$, and
$\bs{Y}_{3\bar{\mathcal{C}}}$. It is not clear why one should prefer either the rows or the columns to perform a factor analysis first,
then perform a respective regression on the columns or rows.

We can frame the estimation of
$\bs{Z}_{2}\bs{\alpha}_{\bar{\mathcal{C}}}$ as a
matrix imputation problem where the missing values are a
submatrix. That is, we try to estimate
$\bs{Z}_{2}\bs{\alpha}_{\bar{\mathcal{C}}}$ given only
the values of $\bs{Y}_{2\mathcal{C}}$,
$\bs{Y}_{3\mathcal{C}}$, and
$\bs{Y}_{3\bar{\mathcal{C}}}$.  In the context of matrix
imputation (and not removing unwanted variation), \citet{owen2016bi},
generalizing the methods of \citet{owen2009bi} to the case of
heteroscedastic noise, suggest that after applying a factor analysis
to $\bs{Y}_{3\mathcal{C}}$, one use the estimates
\begin{align}
  \hat{\bs{Z}}_2 &:= \bs{Y}_{2\mathcal{C}} \hat{\bs{\Sigma}}_{\mathcal{C}}^{-1} \hat{\bs{\alpha}}_{\mathcal{C}}^{\intercal}(\hat{\bs{\alpha}}_{\mathcal{C}}\hat{\bs{\Sigma}}_{\mathcal{C}}^{-1}\hat{\bs{\alpha}}_{\mathcal{C}}^{\intercal})^{-1},\\
  \hat{\bs{\alpha}}_{\bar{\mathcal{C}}} &:= (\hat{\bs{Z}}_3^{\intercal}\hat{\bs{Z}}_3)^{-1}\hat{\bs{Z}}_3^{\intercal}\bs{Y}_{3\bar{\mathcal{C}}},
\end{align}
and then set
$\widehat{\bs{Z}_{2}\bs{\alpha}_{\bar{\mathcal{C}}}} =
\hat{\bs{Z}}_2\hat{\bs{\alpha}}_{\bar{\mathcal{C}}}$.
This corresponds to a factor analysis followed by two regressions
followed by an imputation step.  Following the theme of this paper, we
would add an additional step and estimate
$\bs{\beta}_{2\bar{\mathcal{C}}}$ with
\begin{align}
\hat{\bs{\beta}}_{2\bar{\mathcal{C}}} = \bs{R}_{22}^{-1}(\bs{Y}_{2\bar{\mathcal{C}}} - \hat{\bs{Z}}_2\hat{\bs{\alpha}}_{\bar{\mathcal{C}}}).
\end{align}

In this section, we prove that this estimation procedure, presented in
Procedure \ref{algorithm:ruv3}, is a version of both RUV2 (Section
\ref{section:ruv2.connection}) and RUV4 (Section
\ref{section:ruv4.connection}). Indeed, it is the only such procedure
that is both a version of RUV2 and RUV4 (Section
\ref{section:ruv24iff3}). As such, we call it RUV3 (also presented
pictorially in Figure \ref{fig:ruv.rep}). Like RUV2 and RUV4, RUV3 is
a class of methods indexed by the factor analysis used. We sometimes
explicitly denote this indexing by $\text{RUV3}(\mathcal{F})$.

\begin{algorithm}
  \floatname{algorithm}{Procedure}
  \caption{RUV3}
  \label{algorithm:ruv3}
  \begin{algorithmic}[1]
    \STATE Perform factor analysis on $\bs{Y}_{3\mathcal{C}}$ to
    obtain estimates of $\bs{Z}_3$,
    $\bs{\alpha}_{\mathcal{C}}$ and
    $\bs{\Sigma}_{\mathcal{C}}$.
    \STATE Regress $\bs{Y}_{2\mathcal{C}}$ on
    $\hat{\bs{\alpha}}_{\mathcal{C}}$ to obtain an estimate of
    $\bs{Z}_2$ and regress $\bs{Y}_{3\bar{\mathcal{C}}}$
    on $\hat{\bs{Z}}_3$ to obtain estimates of
    $\bs{\alpha}_{\bar{\mathcal{C}}}$ and
    $\bs{\Sigma}_{\bar{\mathcal{C}}}$. That is
    \begin{align}
      \label{equation:ruv3.z2}\hat{\bs{Z}}_2 &:= \bs{Y}_{2\mathcal{C}} \hat{\bs{\Sigma}}_{\mathcal{C}}^{-1} \hat{\bs{\alpha}}_{\mathcal{C}}^{\intercal}(\hat{\bs{\alpha}}_{\mathcal{C}}\hat{\bs{\Sigma}}_{\mathcal{C}}^{-1}\hat{\bs{\alpha}}_{\mathcal{C}}^{\intercal})^{-1},\\
      \label{equation:ruv3.alpha}\hat{\bs{\alpha}}_{\bar{\mathcal{C}}} &:= (\hat{\bs{Z}}_3^{\intercal}\hat{\bs{Z}}_3)^{-1}\hat{\bs{Z}}_3^{\intercal}\bs{Y}_{3\bar{\mathcal{C}}},\\
      \hat{\bs{\Sigma}}_{\bar{\mathcal{C}}} &:= \diag\left[(\bs{Y}_{3\bar{\mathcal{C}}} - \hat{\bs{Z}}_{3}\hat{\bs{\alpha}}_{\bar{\mathcal{C}}})^{\intercal}(\bs{Y}_{3\bar{\mathcal{C}}} - \hat{\bs{Z}}_{3}\hat{\bs{\alpha}}_{\bar{\mathcal{C}}})\right] / (n - k - q).
    \end{align}
    \STATE Estimate $\bs{\beta}_2$ by
    \begin{align}
      \hat{\bs{\beta}}_2 = \bs{R}_{22}^{-1}(\bs{Y}_{2\bar{\mathcal{C}}} - \hat{\bs{Z}}_2\hat{\bs{\alpha}}_{\bar{\mathcal{C}}}).
    \end{align}
  \end{algorithmic}
\end{algorithm}

\subsection{Connection to RUV4}
\label{section:ruv4.connection}
The astute reader will have noticed that \eqref{equation:ruv3.z2} is
the same as \eqref{equation:cate.z2.est}. RUV3 can be viewed as a
version of RUV4 with a particular factor analysis. Specifically, any
factor analysis applied during RUV4 of the following form will result
in RUV3:
\begin{align}
\{\hat{\bs{\Sigma}}(\bs{Y}_{3\mathcal{C}},\bs{Y}_{3\bar{\mathcal{C}}})&, \hat{\bs{Z}}_3(\bs{Y}_{3\mathcal{C}},\bs{Y}_{3\bar{\mathcal{C}}}), \hat{\bs{\alpha}}(\bs{Y}_{3\mathcal{C}},\bs{Y}_{3\bar{\mathcal{C}}})\} \text{ such that }\\
\label{equation:sigmahat.ruv4}\hat{\bs{\Sigma}}_{\mathcal{C}}(\bs{Y}_{3\mathcal{C}},\bs{Y}_{3\bar{\mathcal{C}}})& = \hat{\bs{\Sigma}}_{\mathcal{C}}(\bs{Y}_{3\mathcal{C}})\\
\hat{\bs{Z}}_3(\bs{Y}_{3\mathcal{C}},\bs{Y}_{3\bar{\mathcal{C}}}) &= \hat{\bs{Z}}_3(\bs{Y}_{3\mathcal{C}}), \text{ and}\\
\hat{\bs{\alpha}}(\bs{Y}_{3\mathcal{C}},\bs{Y}_{3\bar{\mathcal{C}}}) &= (\hat{\bs{\alpha}}_{\mathcal{C}}(\bs{Y}_{3\mathcal{C}}), [\hat{\bs{Z}}_3^{\intercal}\hat{\bs{Z}}_3]^{-1}\hat{\bs{Z}}_3^{\intercal}\bs{Y}_{3\bar{\mathcal{C}}}).\label{equation:alphahat.ruv4}
\end{align}
Or, more simply, RUV4 is equal to RUV3 if, in RUV4 one uses any
factor analysis such that
$\hat{\bs{\alpha}}_{\bar{\mathcal{C}}} =
(\hat{\bs{Z}}_3^{\intercal}\hat{\bs{Z}}_3)^{-1}\hat{\bs{Z}}_3^{\intercal}\bs{Y}_{3\bar{\mathcal{C}}}$
and $\hat{\bs{\Sigma}}_{\mathcal{C}}$,
$\hat{\bs{Z}}_3$, and
$\hat{\bs{\alpha}}_{\mathcal{C}}$ are functions of
$\bs{Y}_{3\mathcal{C}}$ but \emph{not}
$\bs{Y}_{3\bar{\mathcal{C}}}$

Actually, using a truncated SVD on
$(\bs{Y}_{3\mathcal{C}},\bs{Y}_{3\bar{\mathcal{C}}})$ (equations
\eqref{equation:tsvd.z} to \eqref{equation:tsvd.sigma}) results in the
equalities \eqref{equation:sigmahat.ruv4} to
\eqref{equation:alphahat.ruv4} if one ignores the functional
dependencies. That is, using the truncated SVD on
$(\bs{Y}_{3\mathcal{C}},\bs{Y}_{3\bar{\mathcal{C}}})$ one can easily
prove the relationships
\begin{align}
\hat{\bs{Z}}_3(\bs{Y}_{3\mathcal{C}},\bs{Y}_{3\bar{\mathcal{C}}}) &= \hat{\bs{Z}}_3(\bs{Y}_{3\mathcal{C}}, \bs{Y}_{3\bar{\mathcal{C}}}),\\
\hat{\bs{\alpha}}(\bs{Y}_{3\mathcal{C}},\bs{Y}_{3\bar{\mathcal{C}}}) &= (\hat{\bs{\alpha}}_{\mathcal{C}}(\bs{Y}_{3\mathcal{C}}, \bs{Y}_{3\bar{\mathcal{C}}}), [\hat{\bs{Z}}_3^{\intercal}\hat{\bs{Z}}_3]^{-1}\hat{\bs{Z}}_3^{\intercal}\bs{Y}_{3\bar{\mathcal{C}}}).
\end{align}
The main difference, then, between RUV3 and RUV4 as implemented in the
\texttt{ruv} R package \citep{bartsch2015ruv} is that in RUV3
$\hat{\bs{Z}}_3$ has a functional dependence only on
$\bs{Y}_{3\mathcal{C}}$ and \emph{not}
$\bs{Y}_{3\bar{\mathcal{C}}}$.

\subsection{Connection to RUV2}
\label{section:ruv2.connection}
Similar to the relationship between RUV3 and RUV4, RUV3 may also be
viewed as a version of RUV2 with a particular factor
analysis. Specifically any factor analysis applied during RUV2 of the
following form will result in RUV3:
\begin{align}
\{\hat{\bs{\Sigma}}_{\mathcal{C}}(\bs{Y}_{2\mathcal{C}},\bs{Y}_{3\mathcal{C}})&, \hat{\bs{Z}}(\bs{Y}_{2\mathcal{C}},\bs{Y}_{3\mathcal{C}}), \hat{\bs{\alpha}}_{\mathcal{C}}(\bs{Y}_{2\mathcal{C}},\bs{Y}_{3\mathcal{C}})\} \text{ such that }\\
\label{equation:sigmahat.ruv2}\hat{\bs{\Sigma}}_{\mathcal{C}}(\bs{Y}_{2\mathcal{C}},\bs{Y}_{3\mathcal{C}})& = \hat{\bs{\Sigma}}_{\mathcal{C}}(\bs{Y}_{3\mathcal{C}})\\
\hat{\bs{\alpha}}_{\mathcal{C}}(\bs{Y}_{2\mathcal{C}},\bs{Y}_{3\mathcal{C}}) &= \hat{\bs{\alpha}}_{\mathcal{C}}(\bs{Y}_{3\mathcal{C}}), \text{ and}\\
\label{equation:ruv2.gls}
\hat{\bs{Z}}(\bs{Y}_{2\mathcal{C}},\bs{Y}_{3\mathcal{C}}) &=
\left(
\begin{array}{c}
\bs{Y}_{2\mathcal{C}} \hat{\bs{\Sigma}}_{\mathcal{C}}^{-1} \hat{\bs{\alpha}}_{\mathcal{C}}^{\intercal}(\hat{\bs{\alpha}}_{\mathcal{C}}\hat{\bs{\Sigma}}_{\mathcal{C}}^{-1}\hat{\bs{\alpha}}_{\mathcal{C}}^{\intercal})^{-1}\\
\hat{\bs{Z}}_3(\bs{Y}_{3\mathcal{C}})
\end{array}
\right).
\end{align}
In simpler terms, RUV2 is the same as RUV3 if, in RUV2 one uses any
factor analysis such that $\hat{\bs{Z}}_2 =
\bs{Y}_{2\mathcal{C}}
\hat{\bs{\Sigma}}_{\mathcal{C}}^{-1}
\hat{\bs{\alpha}}_{\mathcal{C}}^{\intercal}(\hat{\bs{\alpha}}_{\mathcal{C}}\hat{\bs{\Sigma}}_{\mathcal{C}}^{-1}\hat{\bs{\alpha}}_{\mathcal{C}}^{\intercal})^{-1}$
and $\hat{\bs{\alpha}}_{\mathcal{C}}$,
$\hat{\bs{Z}}_3$, and
$\bs{\bs{\Sigma}}_{\mathcal{C}}$ are functions of
$\bs{Y}_{3\mathcal{C}}$ but not
$\bs{Y}_{2\mathcal{C}}$.

Similar to Section \ref{section:ruv4.connection}, using the truncated
SVD on ${\bs{Y}_{2\mathcal{C}}\choose\bs{Y}_{3\mathcal{C}}}$ ---
assuming homoscedastic variances rather than heteroscedastic variances
--- will result in equalities \eqref{equation:sigmahat.ruv2} to
\eqref{equation:ruv2.gls} except for the functional dependencies. That
is, when using the truncated SVD on
${\bs{Y}_{2\mathcal{C}}\choose\bs{Y}_{3\mathcal{C}}}$ one can show
that the following relationships hold:
\begin{align}
\hat{\bs{\alpha}}_{\mathcal{C}}(\bs{Y}_{2\mathcal{C}},\bs{Y}_{3\mathcal{C}}) &= \hat{\bs{\alpha}}_{\mathcal{C}}(\bs{Y}_{2\mathcal{C}}, \bs{Y}_{3\mathcal{C}}), \text{ and}\\
\hat{\bs{Z}}(\bs{Y}_{2\mathcal{C}},\bs{Y}_{3\mathcal{C}}) &=
\left(
\begin{array}{c}
\bs{Y}_{2\mathcal{C}}\hat{\bs{\alpha}}_{\mathcal{C}}^{\intercal}(\hat{\bs{\alpha}}_{\mathcal{C}}\hat{\bs{\alpha}}_{\mathcal{C}}^{\intercal})^{-1}\\
\hat{\bs{Z}}_3(\bs{Y}_{2\mathcal{C}},\bs{Y}_{3\mathcal{C}})
\end{array}
\right).
\end{align}
The main difference, then, between RUV2 and RUV3 is that in RUV3
$\hat{\bs{\alpha}}_{\mathcal{C}}$ has a functional dependence
only on $\bs{Y}_{3\mathcal{C}}$ and \emph{not}
$\bs{Y}_{2\mathcal{C}}$.

One apparent disadvantage to using the original RUV2 pipeline
(Procedure \ref{algorithm:ruv2.gag} or \ref{algorithm:ruv2.new}), is
that after the factor analysis of step 1, one already has estimates of
$\bs{\alpha}_{\mathcal{C}}$ and
$\bs{\Sigma}_{\mathcal{C}}$, yet one ignores these estimates
and re-estimates them in step 2 with
$\hat{\bs{\alpha}}_{\mathcal{C}} =
(\hat{\bs{Z}}_3^{\intercal}\hat{\bs{Z}}_3)^{-1}\hat{\bs{Z}}_3^{\intercal}\hat{\bs{Y}}_{3C}$
and $\hat{\bs{\Sigma}}_{\mathcal{C}} =
\diag[(\bs{Y}_{3\mathcal{C}} -
\hat{\bs{Z}}_3\hat{\bs{\alpha}}_{\mathcal{C}})^{\intercal}(\bs{Y}_{3\mathcal{C}}
- \hat{\bs{Z}}_3\hat{\bs{\alpha}}_{\mathcal{C}})] / (n
- q)$. In Procedures \ref{algorithm:ruv2.gag} or
\ref{algorithm:ruv2.new}, the estimates from step 1 are in general
different from the estimates from step 2 and it is not at all clear
which estimates one should prefer. RUV3 obviates this problem by the
constrained factor analysis \eqref{equation:ruv2.gls}.

\subsection{RUV2 and RUV4 if and only if RUV3}
\label{section:ruv24iff3}
We have shown in Sections \ref{section:ruv4.connection} and
\ref{section:ruv2.connection} that RUV3 can be considered a variant of
both RUV2 and RUV4. But the converse is easily proved to also be true.
\begin{theorem}
  Suppose a procedure is both a version of RUV4 (Procedure
  \ref{algorithm:ruv4}) and RUV2 (Procedure
  \ref{algorithm:ruv2.new}). Then it is also a version of RUV3
  (Procedure \ref{algorithm:ruv3}).
\end{theorem}
\begin{proof}
  If the procedure is a version of RUV2, then
  \eqref{equation:alphahat.ruv2.rotate} holds. But if the procedure is
  a version of RUV4, then \eqref{equation:cate.z2.est} holds. These
  are two properties of RUV3 (equations \eqref{equation:ruv3.z2} and
  \eqref{equation:ruv3.alpha}).

  It remains to show that $\hat{\bs{Z}}_3$ and
  $\hat{\bs{\alpha}}_{\mathcal{C}}$ are functions only of
  $\bs{Y}_{3\mathcal{C}}$. But this is clear since if the
  procedure is RUV4, these quantities are functions only of
  $\bs{Y}_{3\mathcal{C}}$ and
  $\bs{Y}_{3\bar{\mathcal{C}}}$, while if the procedure is
  RUV2 these quantities are functions only of
  $\bs{Y}_{2\mathcal{C}}$ and
  $\bs{Y}_{3\mathcal{C}}$. Since the procedure is both RUV2
  and RUV4, this necessarily implies that these quantities are
  functions only of $\bs{Y}_{3\mathcal{C}}$.
\end{proof}
To summarize, RUV3 can be viewed as both a version of RUV2 and a version of RUV4
and if a procedure is both a version of RUV2 and a version of RUV4 then it is a
version of RUV3.

\section{A more general framework:\ RUV*}
\label{section:ruv5.full}
A key insight that arises from our unification of RUV2 and RUV4 (and
RUV3) into a single framework is that they share a common goal:
estimation of $\bs{Z}_2 \bs{\alpha}_{\bar{\mathcal{C}}}$, which
represents the combined effects of all sources of unwanted variation
on $\bs{Y}_{2\bar{\mathcal{C}}}$.

This insight suggests a more general approach to the problem:\ any
matrix imputation procedure could be used to estimate
$\bs{Z}_2 \bs{\alpha}_{\bar{\mathcal{C}}}$ --- RUV2, RUV3, and RUV4 are
just three versions that rely heavily on linear associations between
submatrices. Indeed, we need not even assume a factor model for the
form of the unwanted variation. And we can further incorporate
uncertainty in the estimates. In this section we develop these ideas
to provide a more general framework for removing unwanted variation,
which we call RUV*.

\subsection{More general approaches to matrix imputation}

To generalize RUV methods to allow for more general approaches to
matrix imputation we generalize \eqref{equation:matrix.form} to
\begin{align}
\label{equation:general.model}
\left(
\begin{array}{cc}
\bs{Y}_{2\mathcal{C}} & \bs{Y}_{2\bar{\mathcal{C}}}\\
\bs{Y}_{3\mathcal{C}} & \bs{Y}_{3\bar{\mathcal{C}}}
\end{array}
\right)
=
\left(
\begin{array}{cc}
\bs{\Omega}(\bs{\phi})_{2\mathcal{C}} & \bs{\Omega}(\bs{\phi})_{2\bar{\mathcal{C}}}\\
\bs{\Omega}(\bs{\phi})_{3\mathcal{C}} & \bs{\Omega}(\bs{\phi})_{3\bar{\mathcal{C}}}
\end{array}
\right)
 +
\left(
\begin{array}{cc}
\bs{0} & \bs{R}_{22}\bs{\beta}_2\\
\bs{0} & \bs{0}
\end{array}
\right) +
\bs{E},
\end{align}
where $\bs{\Omega}$ is the unwanted variation parameterized by some
$\bs{\phi}$. When the unwanted variation is represented by a factor
model, we have that $\bs{\phi} = \{\bs{Z}, \bs{\alpha}\}$ and
$\bs{\Omega}(\bs{\phi}) = \bs{Z}\bs{\alpha}$.

The simplest version of RUV* fits this model in two steps:
\begin{enumerate}
\item Use any appropriate
procedure to estimate
$\bs{\Omega}_{2\bar{\mathcal{C}}}(\bs{\phi})$ given
$\{\bs{Y}_{2\mathcal{C}}, \bs{Y}_{3\mathcal{C}},
\bs{Y}_{3\bar{\mathcal{C}}}\}$;
\item Estimate $\bs{\beta}_2$ by
\begin{align}
\bs{R}_{22}^{-1}(\bs{Y}_{2\bar{\mathcal{C}}} - \bs{\Omega}_{2\bar{\mathcal{C}}}(\hat{\bs{\phi}})).
\end{align}
\end{enumerate}

This idea, and its relationship with other RUV approaches, are
illustrated in Figure \ref{fig:ruv3.joint}.  Rather than restrict
factors to be estimated via linear regression, RUV* allows any
imputation procedure to be used to estimate
$\bs{\Omega}_{2\bar{\mathcal{C}}}(\bs{\phi})$.  This opens up a large
literature on matrix imputation for use in removing unwanted variation
with control genes \citep[for example]{hoff2007model,
  allen2010transposable, candes2010matrix, stekhoven2012missforest,
  van2012flexible, josse2016denoiser}.

\begin{figure}
\begin{center}
\begin{tikzpicture}
\node at (0, 2.1){\textbf{RUV*}};
\tikzstyle{every path}=[very thick];
\path (-1.3, 1.3) node (y2c) {$\bs{Y}_{2\mathcal{C}}$}
(-1.3, -1.3) node (y3c) {$\bs{Y}_{3\mathcal{C}}$}
(1.3, -1.3) node(y3cbar) {$\bs{Y}_{3\bar{\mathcal{C}}}$}
(1.3, 1.3) node[rounded corners, draw](y2cbar) {$\bs{Y}_{2\bar{\mathcal{C}}}$};
\draw (-1.8, 1.3) .. controls (-1.8, 1.8) .. (-1.3, 1.8) -- (1.8, -1.3) .. controls (1.8, -1.8) .. (1.3, -1.8) -- (-1.3, -1.8) .. controls(-1.8, -1.8) .. (-1.8, -1.3) -- (-1.8, 1.3);
\coordinate (meetup) at (0.25, 0.25);
\path[->] (meetup) edge[below right] node[scale = 0.7]{Impute} (y2cbar);
\end{tikzpicture}
\hspace{1cm}
\begin{tikzpicture}
\node at (0.5, 2.1){\textbf{RUV Relationships}};
\tikzstyle{every path}=[very thick];
\draw (0, -0.5) circle (0.8) (0, -1) node [text=black, scale = 0.7] {RUV4};
\draw (0, 0.5) circle (0.8) (0, 1) node [text=black, scale = 0.7] {RUV2};
\node[align = center, scale = 0.7](fa) at (0, 0) {RUV3};
\draw (.5, 0) circle (1.8) (1.5, 0) node [text=black, scale = 0.7] {RUV*};
\end{tikzpicture}
\hspace{1cm}
\begin{tikzpicture}
\tikzstyle{every path}=[very thick];
\node at (0, 2.1){\textbf{RUVfun}};
\path (-1.3, 1.3) node[rounded corners, draw](y2c) {$\bs{Y}_{2\mathcal{C}}$}
(-1.3, -1.3) node (y3c) {$\bs{Y}_{3\mathcal{C}}$}
(1.3, -1.3) node (y3cbar) {$\bs{Y}_{3\bar{\mathcal{C}}}$}
(1.3, 1.3) node[rounded corners, draw](y2cbar) {$\bs{Y}_{2\bar{\mathcal{C}}}$};
\coordinate (meetup) at (0.6, 0.6);
\path[->] (y3c) edge[right] node[scale = 0.7, align = left]{Train\\Any\\Function} (y2c);
\draw (y3cbar) .. controls(0.6, 0) .. (meetup);
\draw (y2c) .. controls(0, 0.6) .. (meetup);
\draw[->] (meetup) -- (y2cbar);
\node[scale = 0.7, right](impute) at (0.6, 0.1){Impute};
\node[draw = black, fit = (y3c) (y3cbar), rounded corners, inner sep=0ex]{};
\node[scale = 0.65, align = center](fa) at (0, -1.3){Factor\\Analysis};
\end{tikzpicture}
\caption{The left panel contains a pictorial representation of our generalized framework for removing unwanted variation, RUV*, in terms of matrix imputation. The center panel is a Venn diagram representing the
  relationships between RUV2, RUV3, RUV4, and RUV*. The right panel is
  a pictorial representation of RUVfun from
  \citet{gagnon2013removing}}
\label{fig:ruv3.joint}
\end{center}
\end{figure}
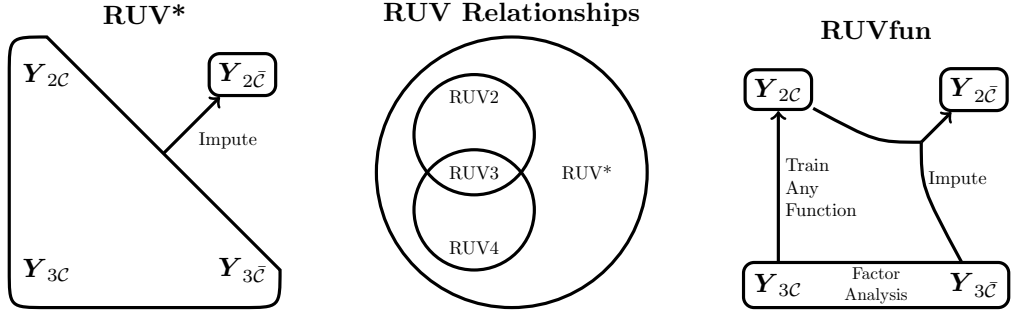

\subsection{Relationship to RUVfun}

\citet{gagnon2013removing} describe a general framework they call
RUVfun, for RUV-functional. In our notation, and within the rotated
model framework of Section \ref{subsection:rotate}, RUVfun may be
described as
\begin{enumerate}[noitemsep]
\item Perform factor analysis on
  $(\bs{Y}_{3\mathcal{C}},
  \bs{Y}_{3\bar{\mathcal{C}}})$
  to obtain an estimate $\hat{\bs{\alpha}}$,
\item Let $\hat{\bs{\alpha}}_j$ denote the $j$th column of
  $(\hat{\bs{\alpha}}_{\mathcal{C}},
  \hat{\bs{\alpha}}_{\bar{\mathcal{C}}})$
  and $\bs{y}_{j}$ denote the $j$th column of
  $(\bs{Y}_{2\mathcal{C}},
  \bs{Y}_{2\bar{\mathcal{C}}})$.
  Train a function $f$ using responses $\bs{y}_{j}$ and
  predictors $\hat{\bs{\alpha}}_{j}$ for $j = 1,\ldots,m$
  (recall, we have $m$ control genes). That is, fit
  \begin{align}
    \bs{y}_{j} \approx f(\hat{\bs{\alpha}}_{j}), \text{ for } j = 1,\ldots, m,
  \end{align}
and call the resulting predictor $\hat{f}$.
\item Estimate $\bs{Y}_{2\bar{\mathcal{C}}}$ using predictors
  $\hat{\bs{\alpha}}_{\bar{\mathcal{C}}}$. That is,
  \begin{align}
    \hat{\bs{y}}_{j} = \hat{f}(\hat{\bs{\alpha}}_{j}), \text{ for } j = m + 1,\ldots, p.
  \end{align}
\item Estimate $\bs{\beta}_{2\bar{\mathcal{C}}}$ with
  $\bs{R}_{22}^{-1}(\bs{Y}_{2\bar{\mathcal{C}}} -
  \hat{\bs{Y}}_{2\bar{\mathcal{C}}})$.
\end{enumerate}

This is the way it is presented in Section 3.8.2 of
\citet{gagnon2013removing}, but typically they take the factor
analysis step to mean just setting
$\hat{\bs{\alpha}} = (\bs{Y}_{3\mathcal{C}},
\bs{Y}_{3\bar{\mathcal{C}}})$.
A pictorial representation of RUVfun is presented in the third panel
of Figure \ref{fig:ruv3.joint}. The only difference between the RUVfun
diagram in Figure \ref{fig:ruv3.joint} and the RUV4 diagram in Figure
\ref{fig:ruv.rep} is that ``regression'' was changed to ``train any
function''.

RUVfun, though more general than RUV4, is a special case of RUV*. And
RUV* is more general:\ for example, RUV2 is a version of RUV* but not
of RUVfun.  Indeed, RUV* generalizes RUVfun in three key ways.  First,
RUVfun uses only one column of $\hat{\bs{\alpha}}$ to estimate one
column of $\bs{Y}_{2\bar{\mathcal{C}}}$ while RUV* allows for joint
estimation of $\bs{Y}_{2\bar{\mathcal{C}}}$. Second, RUVfun assumes
that each column of the rotated $\bs{Y}$ matrix is independent and
identically distributed \citep[page 41]{gagnon2013removing} while RUV*
does not. Indeed, some matrix imputation approaches use column
covariances to great effect \citep{allen2010transposable}. Third,
RUVfun uses only $\bs{Y}_{3\mathcal{C}}$ to train the prediction
function, whereas RUV* can use all elements in the rotated $\bs{Y}$
matrix.

\subsection{Incorporating uncertainty in estimated unwanted variation}
\label{section:ruvb}

Like previous RUV methods, the second step of RUV* described above
treats the estimate of $\bs{\Omega}_{2\bar{\mathcal{C}}}(\bs{\phi})$
from the first step as if it were ``known'' without error. Here we
generalize this, using Bayesian ideas to propagate the uncertainty
from the first step to the next.

Although the use of Bayesian methods in this context is not new
\citep{stegle2008accounting, stegle2010bayesian, fusi2012joint,
  stegle2012using}, our development here shares one of the great
advantages of the RUV methods:\ {\it modularity}.  That is, RUV
methods separate the analysis into smaller self-contained steps:\ the
factor analysis step and the regression step. Modularity is widely
used in many fields:\ mathematicians modularize results using
theorems, lemmas and corollaries; computer scientists modularize code
using functions and classes. Modularity has many benefits, including:\
(1) it is easier to conceptualize an approach if it is broken into
small simple steps, (2) it is easier to discover and correct mistakes,
and (3) it is easier to improve an approach by improving specific
steps. These advantages also apply to statistical analysis and methods
development. For example, in RUV if one wishes to use a new method for
factor analysis then this does not require a whole new approach ---
one simply replaces the truncated SVD with the new factor analysis.

To describe this generalized RUV* we
introduce a latent variable $\tilde{\bs{Y}}_{2\bar{\mathcal{C}}}$ and write
\eqref{equation:general.model} as
\begin{align}
\label{equation:matrix}
\left(
\begin{array}{cc}
\bs{Y}_{2\mathcal{C}} & \tilde{\bs{Y}}_{2\bar{\mathcal{C}}}\\
\bs{Y}_{3\mathcal{C}} & \bs{Y}_{3\bar{\mathcal{C}}}
\end{array}
\right)
&=
\bs{\Omega}(\bs{\phi}) + \bs{E},\\
\label{equation:true.second.model}
  \bs{Y}_{2\bar{\mathcal{C}}} &=
  \bs{R}_{22}\bs{\beta}_2 +
  \tilde{\bs{Y}}_{2\bar{\mathcal{C}}}.
\end{align}

Now consider the following two-step procedure:
\begin{enumerate}
\item  Use any appropriate
procedure to obtain a conditional distribution $h(\tilde{\bs{Y}}_{2\bar{\mathcal{C}}}) := p(\tilde{\bs{Y}}_{2\bar{\mathcal{C}}}|\mathcal{Y}_m)$,
where $\mathcal{Y}_m := \{\bs{Y}_{2\mathcal{C}},
\bs{Y}_{3\mathcal{C}}, \bs{Y}_{3\bar{\mathcal{C}}}\}$.
\item Perform inference for $\bs{\beta}_2$ using the likelihood
\begin{align}
L(\bs{\beta}_2) & := p(\bs{Y}_{2\bar{\mathcal{C}}},\mathcal{Y}_m | \bs{\beta}_2) \\
&= p(\mathcal{Y}_m) \int p(\bs{Y}_{2\bar{\mathcal{C}}} | \tilde{\bs{Y}}_{2\bar{\mathcal{C}}},\bs{\beta}_2) p(\tilde{\bs{Y}}_{2\bar{\mathcal{C}}} | \mathcal{Y}_m) \, d\tilde{\bs{Y}}_{2\bar{\mathcal{C}}} \\
&= p(\mathcal{Y}_m) \int \delta( \bs{Y}_{2\bar{\mathcal{C}}}  - \tilde{\bs{Y}}_{2\bar{\mathcal{C}}} - \bs{R}_{22}\bs{\beta}_2) p(\tilde{\bs{Y}}_{2\bar{\mathcal{C}}} | \mathcal{Y}_m) \, d\tilde{\bs{Y}}_{2\bar{\mathcal{C}}} \\
& = p(\mathcal{Y}_m) h(\bs{Y}_{2\bar{\mathcal{C}}} - \bs{R}_{22}\bs{\beta}_2) \\
& \propto h(\bs{Y}_{2\bar{\mathcal{C}}} - \bs{R}_{22}\bs{\beta}_2),
\end{align}
where $\delta(\cdot)$ indicates the Dirac delta function.
\end{enumerate}
Of course, in step 2 one could do classical inference for $\bs{\beta}_2$, or place a prior on $\bs{\beta}_2$ and perform Bayesian inference.

This procedure requires an analytic form for the conditional distribution $h$. An alternative is to assume that we can sample from this conditional
distribution, which yields a
convenient sample-based (or ``multiple imputation") RUV* algorithm.
\begin{enumerate}
\item  Use any appropriate
procedure to obtain samples  $\tilde{\bs{Y}}_{2\bar{\mathcal{C}}}^{(1)}, \ldots,
\tilde{\bs{Y}}_{2\bar{\mathcal{C}}}^{(t)}$
from a conditional distribution $p(\tilde{\bs{Y}}_{2\bar{\mathcal{C}}}|\mathcal{Y}_m)$.
\item Approximate the likelihood for $L(\bs{\beta}_2)$ by using the fact that $\hat{\bs{\beta}}_2^{(i)} :=
\bs{R}_{22}^{-1}(\bs{Y}_{2\bar{\mathcal{C}}} -
\tilde{\bs{Y}}_{2\bar{\mathcal{C}}}^{(i)})$ are sampled from a distribution proportional to $L(\bs{\beta}_2)$ [which is guaranteed
to be proper; Appendix \ref{section:proper.posterior}].
\end{enumerate}
For example, here in step 2 we can approximate the likelihood for each element of $\bs{\beta}_2$ by a normal likelihood
\begin{equation}
\label{eq:ruvb.normal.approx}
L(\beta_{2j}) \approx N(\beta_{2j}; \hat{\beta}_{2j}, \hat{s}_j^2)
\end{equation}
where $\hat{\beta}_{2j}$ and $\hat{s}_j$ are respectively the mean and
standard deviation of $\hat{\bs{\beta}}_2^{(i)}$. Alternatively, a $t$
likelihood can be used. Either approach provides an estimate and
standard error for each element of ${\bs\beta}_2$ that accounts for
uncertainty in the estimated unwanted variation. This is in contrast
to the various methods used by the other RUV approaches to estimate
the standard errors which do not account for this uncertainty (Section
\ref{section:variance.estimation}).  Here we use these values to rank
the ``significance'' of genes by the value of
$\hat{\beta}_{2j}/\hat{s}_j$.  They could also be used as inputs to
the empirical Bayes method in \citet{stephens2016false} to obtain
measurements of significance related to false discovery rates.

Other approaches to inference in Step 2 are also possible. For
example, given a specific prior on $\bs{\beta}_{2}$, Bayesian
inference for $\bs{\beta}_{2}$ could be performed by re-weighting
these samples according to this prior distribution (Appendix
\ref{section:approx.post.inf}). This re-weighting yields an
arbitrarily accurate approximation to the posterior distribution
$p(\bs{\beta}_2|\mathcal{Y}_m, \bs{Y}_{2\bar{\mathcal{C}}})$ (Appendix
\ref{section:just.like}). Posterior summaries using this re-weighting
scheme are easy to derive (Appendix \ref{section:post.summ}).

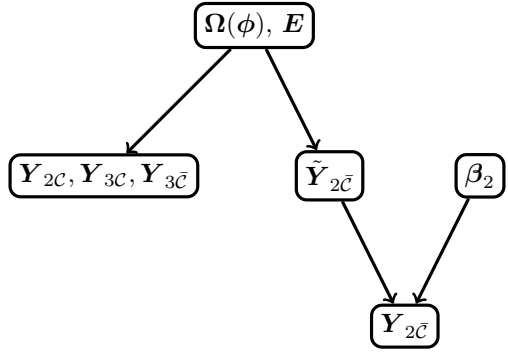
\begin{figure}
\begin{center}
\begin{tikzpicture}
  \tikzstyle{every path}=[very thick];
  \path (-2, 2) node [rounded corners, draw] (omega) {$\bs{\Omega}(\bs{\phi})$, $\bs{E}$}
  (-4, 0) node [rounded corners, draw] (ym) {$\bs{Y}_{2\mathcal{C}}, \bs{Y}_{3\mathcal{C}}, \bs{Y}_{3\bar{\mathcal{C}}}$}
  (-1, 0) node [rounded corners, draw] (y22tilde) {$\tilde{\bs{Y}}_{2\bar{\mathcal{C}}}$}
  (1, 0) node [rounded corners, draw] (beta) {$\bs{\beta}_2$}
  (0, -2) node [rounded corners, draw] (y22) {$\bs{Y}_{2\bar{\mathcal{C}}}$};
  \draw[->] (omega) -- (ym);
  \draw[->] (omega) -- (y22tilde);
  \draw[->] (y22tilde) -- (y22);
  \draw[->] (beta) -- (y22);
\end{tikzpicture}
\end{center}
\caption{Graphical model for the elements in \eqref{equation:matrix} and \eqref{equation:true.second.model}.}
\label{figure:graphical.model}
\end{figure}

\subsection{RUVB:\ Bayesian factor analysis in RUV*}

To illustrate the potential for RUV* to produce new methods for
removing unwanted variation we implemented a version of RUV*, using a
Markov chain Monte Carlo scheme to fit a simple Bayesian Factor
analysis model, and hence perform the sampling-based imputation in
Step 1 of RUV*. See Appendix \ref{section:bfa} for details.  We refer
to this method as RUVB.

\section{Estimating Standard Errors}
\label{section:variance.estimation}

For simplicity we have focused our descriptions of RUV2, RUV3, RUV4,
and RUV* on point estimation for $\bs{\beta}_{2}$. In practice, to be
useful, all of these methods must also provide standard errors for
these estimates.  Several different approaches to this problem exist,
and we have found in empirical comparisons (e.g.\ Section
\ref{section:evaluate} below) that the approach taken can greatly
affect results, particularly calibration of interval estimates. In
this section we therefore briefly review some of these approaches.

The simplest approach is to treat the estimated $\hat{\bs{Z}}$ as the
true value of ${\bs{Z}}$, and then use standard theory from linear
models to estimate the standard errors of the estimated
coefficients. That is, first estimate $\hat{\bs{Z}}$ using any of the
RUV approaches, then regress $\bs{Y}$ on $(\bs{X}, \hat{\bs{Z}})$ and
obtain estimated standard errors (for coefficients of $\bs{X}$) in the
usual way. This is the default option in the \texttt{ruv} R package. The
\texttt{cate} R package implements this (with the
\texttt{nc.var.correction} and \texttt{calibrate} parameters both set
to \texttt{FALSE}), but without the usual degrees of freedom
correction in estimated standard errors. Though asymptotically this
will not matter, we have found that for small sample sizes this can
substantially hurt performance due to downward-biased standard errors.

\citet{gagnon2013removing} noted that the standard errors estimated
using this simple approach can be poorly behaved, essentially because
the $\hat{\bs{Z}}$ are estimated and not known. They suggested
several approaches to calibrating these standard errors using control
genes. The approach that we found to work best in our comparisons (at
least, when there are many control genes --- see Section
\ref{section:sims}) is to multiply the estimated standard errors by a
factor $\lambda$ (i.e.\ set $\tilde{s}_i = \lambda \hat{s}_i$) which is
estimated from control genes by
\begin{align}
\label{eq:gb.var.inflate}
\lambda:=\left(\frac{1}{|\mathcal{C}|}\sum_{j\in \mathcal{C}} \frac{\hat{\beta}_j^2}{\hat{s}_j^2}\right)^{0.5}
\end{align}
where $\hat{\beta_j}$ and $\hat{s_j}$ are the estimated coefficients
and their standard errors (Equation (236) in
\citet{gagnon2013removing}).  In our empirical comparisons below we
refer to this procedure as ``control gene calibration''.
\citet{gagnon2013removing} use heuristic arguments to motivate
\eqref{eq:gb.var.inflate} in the context of studies with just one
covariate of interest. In Appendix \ref{section:catec}, we extend
\eqref{eq:gb.var.inflate} to the case when there is more than one
covariate of interest and formally justify it with maximum likelihood
arguments.

\citet{sun2012multiple} take a different approach to calibration,
which does not use control genes, but is motivated by the assumption
that most genes are null. Specifically they suggest centering the
$t$-statistics $\hat{\beta}_j/\hat{s}_j$ by their median and scaling
them by their median absolute deviation (MAD):
\begin{equation}
\tilde{t}_i = \frac{\hat{\beta}_i / \hat{s}_i - \median\left(\hat{\beta}_1 / \hat{s}_1,\ldots, \hat{\beta}_p / \hat{s}_p\right)}{\MAD\left(\hat{\beta}_1 / \hat{s}_1,\ldots, \hat{\beta}_p / \hat{s}_p\right)}.
\end{equation}
The motivation for this adjustment is that if most genes are null,
then normalizing by robust estimates of the null $t$-statistics'
center and scale will make the null $t$-statistics more closely match
their theoretical distribution. This adjustment of $t$ statistics is
closely connected with the variance calibration
\eqref{eq:gb.var.inflate}. Indeed, if we assume that the median of the
$t$-statistics is approximately zero, then it is effectively
equivalent to variance calibration with
\begin{equation}
\label{eq:mad.se}
\lambda_\text{MAD} := \MAD\left(\frac{\hat{\beta}_1}{\hat{s}_1},\ldots, \frac{\hat{\beta}_p}{\hat{s}_p}\right)
\end{equation}
in place of \eqref{eq:gb.var.inflate}.

In addition to MAD calibration, \citet{wang2015confounder} offer
asymptotic arguments for an additive variance adjustment. This
additive inflation term is particularly important when there are few
control genes, and is specific to the RUV4 estimator (unlike
\eqref{eq:gb.var.inflate} and \eqref{eq:mad.se} which can be applied
to any estimator).

Finally, before development of RUV-like methods, the benefits of using
empirical Bayes variance moderation (EBVM) \citep{smyth2004linear}
were widely recognized in gene expression analyses.  Variance
moderation can be applied in combination with the variance calibration
methods discussed above:\ for example, the {\tt ruv::variance\_adjust}
function in R applies EBVM before applying
\eqref{eq:gb.var.inflate}. EBVM can similarly be incorporated into
other methods. For RUV3 and CATE we can use EBVM either before
or after the generalized least squares (GLS) step of their respective
algorithms (equation \eqref{equation:ruv3.z2} for RUV3 and equation
\eqref{equation:cate.z2.est} for CATE).

\section{Empirical Evaluations}
\label{section:evaluate}
\label{section:sims}

\subsection{Simulation Approach}

We now use simulations based on real data to compare methods (focusing
on methods that use control genes, although the same simulations could
be useful more generally).  In brief, we use random subsets of real
expression data to create ``null data'' that contains real (but
unknown) ``unwanted variation'', and then modify these null data to add
(known) signal. We compare methods in their ability to reliably detect
real signals and avoid spurious signals.  Because they are based on
real data, our simulations involve realistic levels of unwanted
variation. However, they also represent a ``best-case'' scenario in
which treatment labels were randomized with respect to the factors
causing this unwanted variation. They also represent a best case
scenario in that the control genes given to each method are simulated
to be genuinely null. Even in this best-case scenario unwanted
variation is a major issue, and, as we shall see, obtaining well
calibrated inferences is challenging.

In more detail:\ we took the top 1000 expressed genes from the RNA-seq
data on muscle samples from the Genotype Tissue Expression Consortium
\citep{gtex2015}. For each dataset in our simulation study, we
randomly selected $n$ samples ($n = 6$, $10$, $20$, or $40$). We then
randomly assigned half of these samples to be in one group and the
other half to be in a second group. So our design matrix
$\bs{X} \in \mathbb{R}^{n \times 2}$ contains two columns --- a column
of ones and a column that is an indicator for group assignment.

At this point, all genes are theoretically null, as in the datasets of
our introduction (Section \ref{section:introduction}). We used this
``all-null'' scenario as one setting in our simulation studies
(similar to the simulations in \citet{rocke2015excess}). For other
settings, we added signal to a randomly selected proportion of genes
$\pi_1 = 1 - \pi_0$ ($\pi_1 = 0.1$ or $0.5$). To add signal, we first
sampled the effect sizes from a $N(0, 0.8^2)$. The standard deviation,
$0.8$, was chosen by trial and error so that the classification
problem would be neither too easy nor too hard. Let
\begin{align}
a_{j_1},\ldots,a_{j_{\pi_1 p}} \overset{iid}{\sim} N(0, 0.8^2),
\end{align}
be the effect sizes, where $j_{\ell} \in \Omega$, the set of non-null
genes. Then we drew a $\bs{W}$ matrix of the same dimension as our
RNA-seq count matrix $\bs{Z}$ by
\begin{align}
w_{ij_{\ell}}|z_{ij_{\ell}} \sim
\begin{cases}
\text{Binomial}(z_{ij_{\ell}}, 2^{a_{j_{\ell}}x_{i2}}) &\text{if } a_{j_{\ell}} < 0 \text{ and } j_{\ell} \in \Omega,\\
\text{Binomial}(z_{ij_{\ell}}, 2^{-a_{j_{\ell}}(1 - x_{i2})}) &\text{if } a_{j_{\ell}} > 0  \text{ and } j_{\ell} \in \Omega\\
\delta(z_{ij_{\ell}}) &\text{if } j_{\ell} \notin \Omega,
\end{cases}
\end{align}
where $\delta(z_{ij_{\ell}})$ indicates a point mass at
$z_{ij_{\ell}}$. We then used $\bs{W}$ as our new RNA-seq
matrix. Prior to running each method, we took a simple $\log_2$
transformation of the elements in $\bs{W}$ to obtain our $\bs{Y}$
matrix.

To motivate this approach, suppose $z_{ij} \sim \poisson(\lambda_j)$, then
\begin{align}
[w_{ij} | a_{j}, a_{j} < 0, j \in \Omega] &\sim \poisson(2^{a_{j}x_{i2}}\lambda_j)\\
[w_{ij} | a_{j}, a_{j} > 0, j \in \Omega] &\sim \poisson(2^{-a_{j}(1 - x_{i2})}\lambda_j).
\end{align}
Hence,
\begin{align}
E[\log_2(w_{ij}) - \log_2(w_{kj}) | a_{j},\ a_{j} < 0,\ j \in \Omega] &\approx a_{j}x_{i2} - a_{j}x_{k2} = a_j(x_{i2} - x_{k2}), \text{ and}\\
E[\log_2(w_{ij}) - \log_2(w_{kj}) | a_{j},\ a_{j} > 0,\ j \in \Omega] &\approx -a_{j}(1 - x_{i2}) + a_{j}(1 - x_{k2}) = a_j(x_{i2} - x_{k2}).
\end{align}
So $a_j$ is approximately the $\log_2$-fold difference between the two
groups.

\subsection{Summary of Methods Compared}

We compared standard ordinary least squares regression (OLS) against
five other approaches:\ RUV2, RUV3, RUV4, CATE (the GLS variant of
RUV4), and RUVB. We tried each of these effect-estimation methods with
different approaches to variance estimation (Section
\ref{section:variance.estimation}). Specifically, for the non-RUVB
methods we considered:
\begin{itemize}[noitemsep, nolistsep]
\item Variance moderation (EBVM) versus no moderation.
\item Variance calibration (both MAD and control-gene based) vs no calibration.
  \item For RUV3 and CATE:\ EBVM before GLS or after GLS.
\item For CATE:\ additive variance inflation vs no additive variance inflation.
\end{itemize}
Altogether this gave 6 different OLS approaches, 6 RUV2 approaches, 9
RUV3 approaches, 6 RUV4 approaches, and 15 CATE approaches. (We did
not implement CATE with additive variance inflation and EBVM before
GLS because this implementation is not straightforward given the
current \texttt{cate} software.)

For RUVB, we considered four approaches to producing mean and variance
estimates:
\begin{itemize}[noitemsep, nolistsep]
\item Using sample-based posterior summaries (Appendix
  \ref{section:post.summ}).
\item Using the normal approximation to
  the likelihood in Equation \eqref{eq:ruvb.normal.approx}.
\item Using a $t$ approximation to
  the likelihood, replacing Equation
  \eqref{eq:ruvb.normal.approx} with a $t$ density with $n - k - q$
  degrees of freedom.
\item Using a $t$ approximation as above, followed by EBVM.
\end{itemize}
Additional technical considerations are discussed in Appendix
\ref{section:additional.considerations}.

Although the large number of methods considered may seem initially
excessive, we have found that there is often more variation in
performance among different versions of a method than among the
different families of method (RUV2, RUV3, RUV4/CATE, and RUVB). That
is, choices among strategies such as EBVM and variance calibration may
matter as much (or more) in practice as choices between families such
as RUV2 vs RUV4.

\subsection{Comparisons:\ Sensitivity vs Specificity}

To compare methods in their ability to distinguish null and non-null
genes we calculated the area under the receiver operating
characteristic curve (AUC) for each method as the significance
threshold is varied. (Multiplicative variance calibration does not
affect AUC because it does not change the significance rankings among
genes; thus we do not discuss variance calibration methods in this
section.)

The clearest result here is that all the methods consistently
outperform standard OLS (Supplementary Figure
\ref{figure:diff_full}). This emphasizes the benefits of removing
unwanted variation in improving power to detect real effects. For
small sample size comparisons (e.g.\ 3 vs 3) the gains are smaller ---
though still apparent --- presumably because reliably estimating the
unwanted variation is harder for small samples.

A second clear pattern is that the use of variance moderation (EBVM)
consistently improved AUC performance:\ the best-performing method in
each family used EBVM.  As might be expected, these benefits of EBVM
are greatest for smaller sample sizes (Supplementary Figure
\ref{figure:diff_full}).

Compared with these two clear patterns, differences among the
best-performing methods in each family are more subtle.  Figure
\ref{fig:auc.medians} compares the AUC of the best method in each
family with that of RUVB, which performed best overall in this
comparison. (Results are shown for $\pi_0 = 0.5$; results for
$\pi_0 = 0.9$ are similar).  We highlight four main results:
\begin{enumerate}
\item RUVB has the best mean AUC among all methods we explored;
\item RUV4/CATE methods perform less well (relative to RUVB) when
  there are few control genes and the sample size is large;
\item In contrast, RUV2 methods perform less well (relative to RUVB)
  when the sample size is small and there are few control genes;
\item RUV3 performs somewhat stably (relative to RUVB) across the
  sample sizes.
\end{enumerate}

\begin{figure}
\begin{center}
\includegraphics{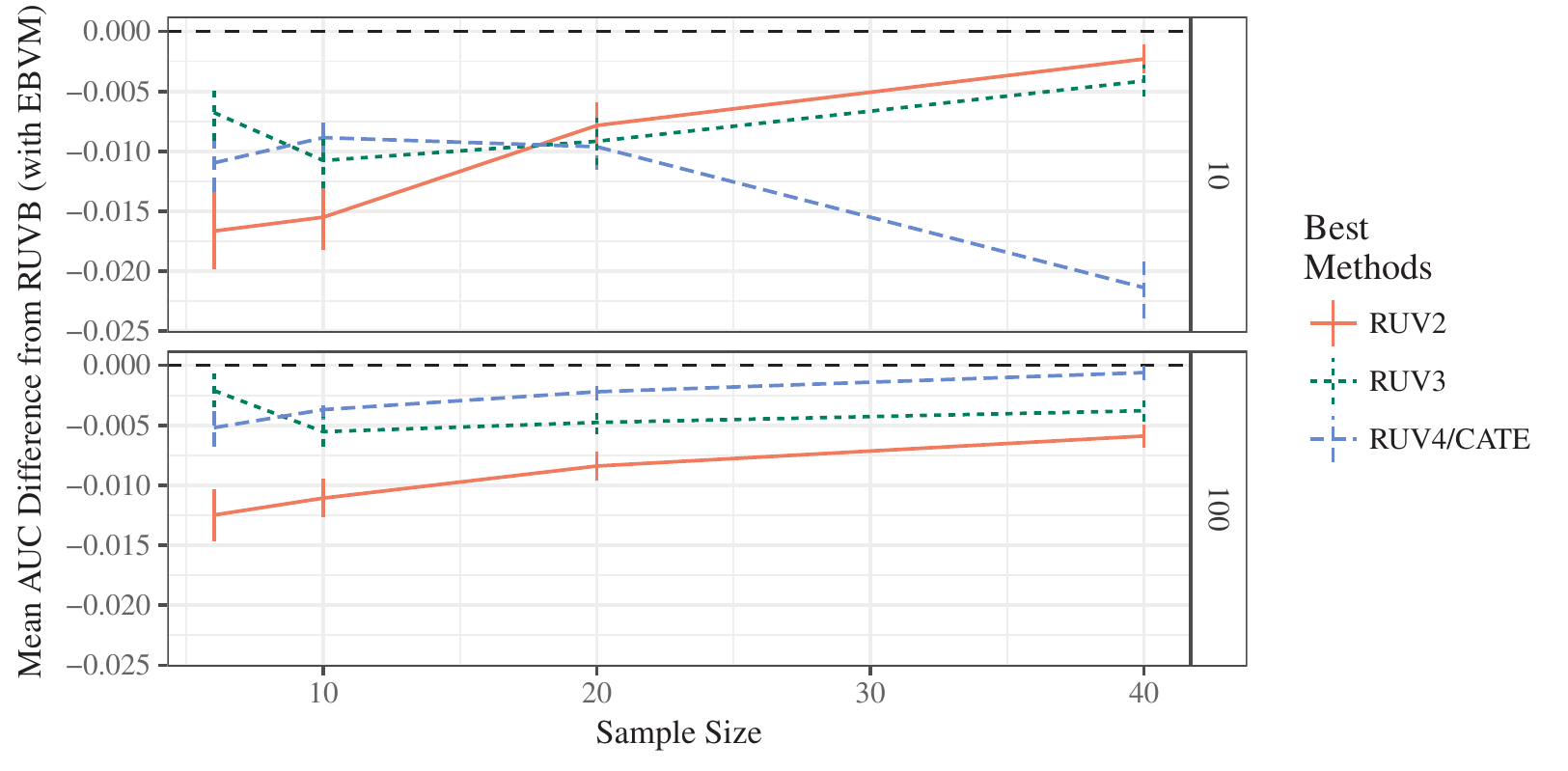}
\end{center}
\caption{Comparison of AUC achieved by best-performing method in each
  family vs RUVB. Each point shows the observed mean difference in
  AUC, with vertical lines indicating 95\% confidence intervals for
  the mean. Results are shown for $\pi_0 = 0.5$ with 10 control genes
  (upper facet) or 100 control genes (lower facet). All results are
  below zero (the dashed horizontal line), indicating superior
  performance of RUVB.}
\label{fig:auc.medians}
\end{figure}

\subsection{Comparisons:\ Calibration}

We also assessed the calibration of methods by examining the empirical
coverage of their nominal $95\%$ confidence intervals for each effect
(based on standard theory for the relevant $t$ distribution in each
case). Because the variance calibration methods can have a strong
effect (on all methods) we initially consider methods without variance
calibration.

We begin by examining ``typical'' coverage for each method in each
scenario by computing the median (across datasets) of the empirical
coverage. We found that, without variance calibration, all method
families except RUV4/CATE were able to achieve satisfactory typical
coverage --- somewhere between 0.94 and 0.97 --- across all scenarios
(Figure \ref{fig:cov.med}a) shows results for $\pi_0 = 0.5$; other
values yielded similar results, not shown).  The best performing
method (in terms of typical coverage) in the RUV4/CATE family was CATE
with only additive variance inflation. However, this method was often
overly conservative in scenarios with few control genes, especially
with larger sample sizes.

Although these median coverage results are encouraging, in practice
having small variation in coverage among datasets is also
important. That is, we would like methods to have near-95\% coverage
in most data-sets, and not only on average. Here the results (Figure
\ref{fig:cov.med}b; Supplementary Figure
\ref{fig:best.boxplots.coverage}) are less encouraging:\ the coverage
of the methods with good typical coverage (median coverage close to
95\%) varied considerably among datasets. This said, variability does
improve for larger sample sizes and more control genes, and in this
case all methods improve noticeably on OLS (Figure \ref{fig:cov.med}b, right facet).
A particular concern is
that, across all these methods, for many datasets, empirical coverage
can be much lower than the nominal goal of 95\%. Such datasets might
be expected to lead to problems with over-identification of
significant null genes (``false positives''), and under-estimation of
false discovery rates when using either frequentist or Empirical Bayes
FDR-related methodology
\citep[e.g.][]{benjamini1995controlling,storey2003positive,stephens2016false}.

To summarize variability in coverage --- as well as any tendency to be
conservative or anti-conservative --- we calculated the proportion of
datasets where the actual coverage deviated substantially from $95\%$,
which we defined as being either less than $90\%$ or more than
$97.5\%$.  Figure \ref{figure:loss.plots} shows these proportions for
each method. Here we have also included methods that use variance
calibration, as the results help highlight the effects of these
calibration methods.  The key findings are:
\begin{enumerate}
\item RUVB (the normal and sample-based versions) has ``balanced''
  errors in coverage:\ its empirical coverage is as likely to be too
  high as too low.
\item MAD calibration tends to produce highly conservative coverage ---
  that is, its coverage is very often much larger than the claimed
  95\%, and seldom much lower. This will tend to reduce false positive
  significant results, but also substantially reduce power to detect
  real effects.  The exception is that when all genes are null
  ($\pi_0=1$), MAD calibration works well for larger sample sizes.
  These results are likely explained partly by non-null genes biasing
  upwards the variance calibration parameter, an issue also noted in
  \citet{sun2012multiple}.
\item Control-gene calibration is often anti-conservative when there
  are few control genes.  However, it can work well when the sample
  size is large and there are many control genes.  Interestingly, with
  few control genes the anti-conservative behavior gets worse as
  sample size increases.
\end{enumerate}

\begin{figure}
\begin{center}
  \includegraphics{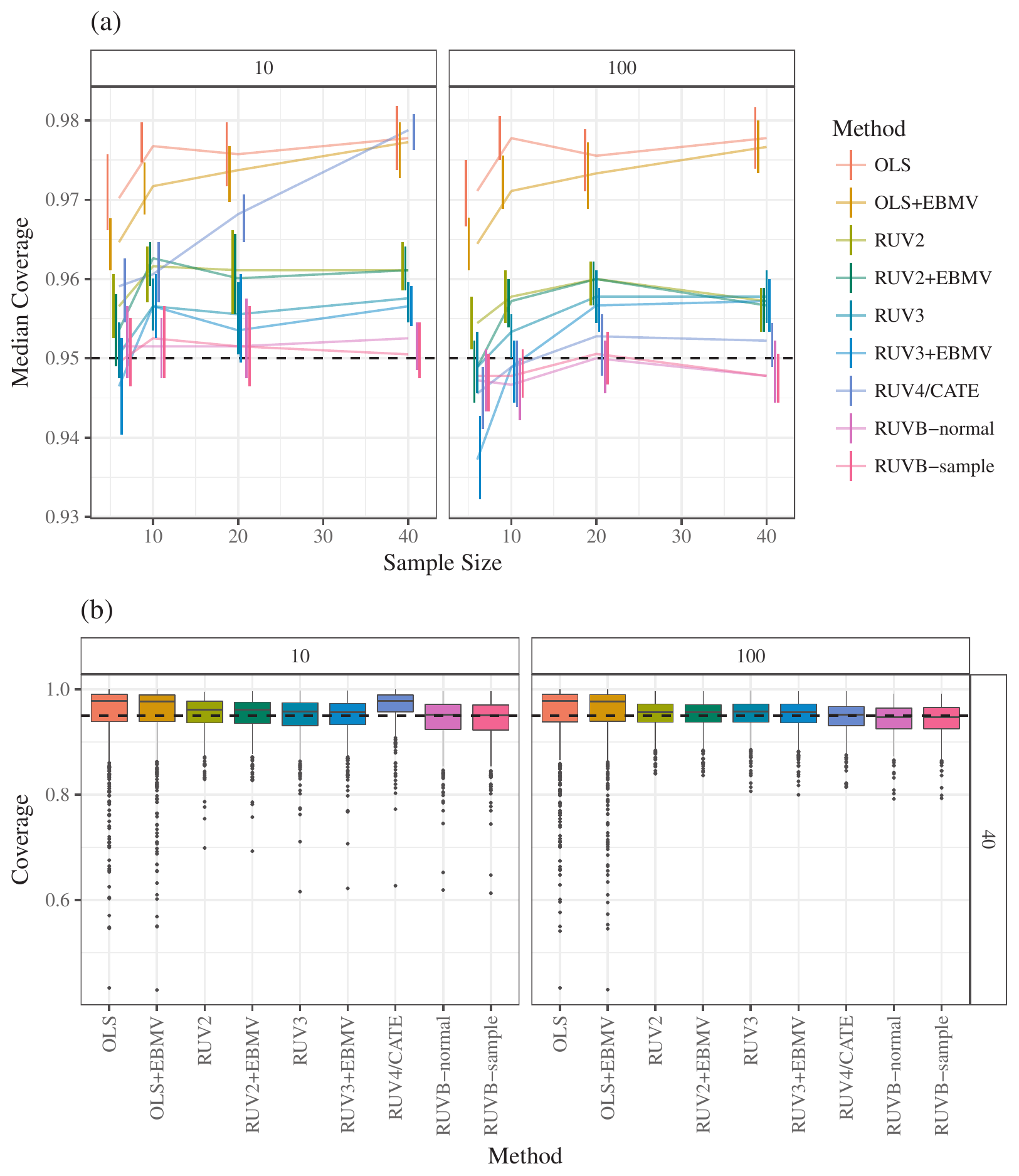}
  \caption{(a) Median coverage for the best performing methods' 95\%
    confidence intervals when $\pi_0 = 0.5$. The vertical lines are
    bootstrap 95\% confidence intervals for the median coverage,
    horizontally dodged to avoid overlap. The horizontal dashed line
    is at 0.95. (b) Boxplots of coverage for the best performing
    methods' 95\% confidence intervals when $\pi_0 = 0.5$ and
    $n = 40$. For both (a) and (b) the left and right facets show
    results for 10 and 100 control genes respectively.}
  \label{fig:cov.med}
\end{center}
\end{figure}

\begin{figure}
\begin{center}
\includegraphics{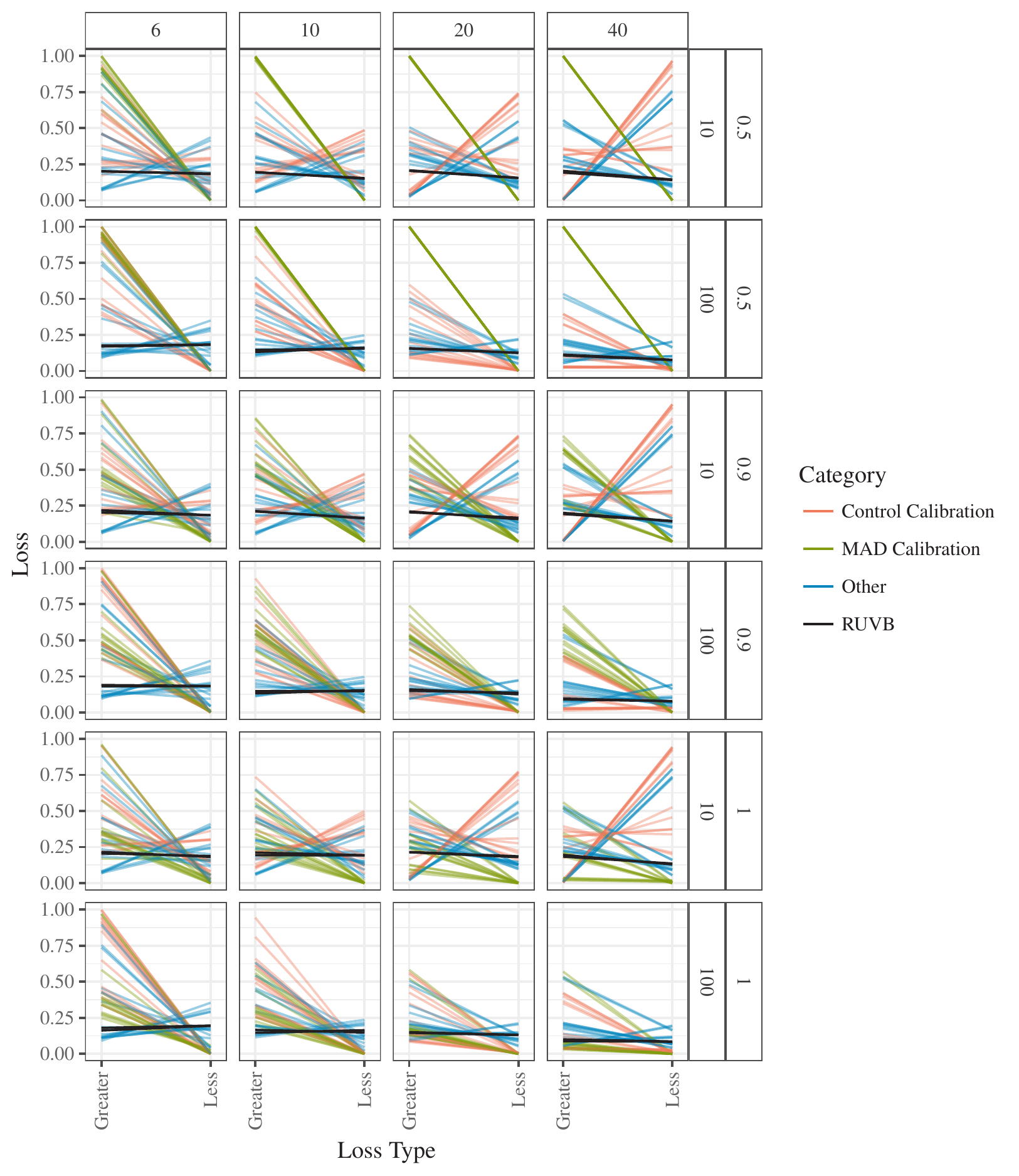}
\caption{Proportion of times the coverage for each method was either
  greater than 0.975 (Greater) or less than 0.9 (Less). The column
  facets distinguish between sample sizes while the row facets
  distinguish between the number of control genes and the proportion
  of genes that are null. The methods are color coded by variance
  calibration method:\ MAD calibrated \eqref{eq:mad.se}, control-gene
  calibrated \eqref{eq:gb.var.inflate}, non-calibrated (other), or the
  sample-based or normal-based RUVB approach.}
\label{figure:loss.plots}
\end{center}
\end{figure}

\subsection{Comparisons:\ Summary}

Our empirical comparisons show that all methods for removing unwanted
variation consistently improve on OLS in terms of ranking non-null
genes vs null genes, with RUVB overall performing best here.  In terms
of calibration, several methods --- including RUVB --- were capable of
providing good ``typical'' calibration across datasets. However,
providing consistently correct calibration remains a challenge even in
our relatively idealized scenarios.  Use of control-gene calibration
can be effective provided sample sizes are sufficiently large and
sufficient control genes are available. In practice the main challenge
here is likely to be specifying a sufficiently-large set of reliable
control genes. Use of MAD calibration can ensure conservative
behavior, but at a potential substantial loss in power to detect real
effects.

\section{Discussion}
\label{section:discussion}

In this paper we developed a framework, RUV*, that both unites and
generalizes different approaches to removing unwanted variation that
use control genes. This unifying framework, which is based on viewing
the problem as a matrix imputation problem, helps clarify connections
between existing methods. In particular we provide conditions under
which two popular methods, RUV2 and RUV4, are equivalent. The RUV*
framework also preserves one of the advantages of existing RUV
approaches --- their modularity --- which facilitates the development of
novel methods based on existing matrix imputation algorithms. At the
same time RUV* extends the RUV methods to allow for incorporating
uncertainty in estimated unwanted variation.  We provide one
illustration of this via RUVB, a version of RUV* based on Bayesian
factor analysis.  In realistic simulations based on real data we found
that RUVB is competitive with existing methods in terms of both power
and calibration, although we also highlighted the challenges of
providing consistently reliable calibration among data sets.

The methods developed in this paper are implemented in the R package
\texttt{vicar} available at
\begin{center}
\href{https://github.com/dcgerard/vicar}{https://github.com/dcgerard/vicar}.
\end{center}
This package contains functions that easily allow analysts to include
their own factor analysis code or (in the case of RUVB) their own
prior specifications into the RUV pipelines. Code and instructions for
reproducing the empirical evaluations in Section
\ref{section:evaluate} are available at
\href{https://github.com/dcgerard/ruvb\_sims}{https://github.com/dcgerard/ruvb\_sims}.

\section{Acknowledgments}
This work was supported by NIH grant HG02585 and by a grant from the
Gordon and Betty Moore Foundation (Grant GBMF \#4559). Much of the
original code for the simulated dataset generation in Section
\ref{section:sims} was based on implementations by Mengyin Lu, to whom
we are indebted. We would also like to express our sincere thanks to
the authors of the \texttt{cate} \citep{wang2015cate}, \texttt{ruv}
\citep{bartsch2015ruv}, and \texttt{sva} \citep{leek2016sva} R
packages for making their work accessible for utilization or
comparison. We also thank the Genotype-Tissue Expression Consortium
\citep{gtex2015} for making their data available for analysis.

\appendix
\section{Appendix}

\subsection{Approximate Posterior Inference}
\label{section:approx.post.inf}

As discussed in Section \ref{section:ruvb}, if we could calculate
$p(\tilde{\bs{Y}}_{2\bar{\mathcal{C}}}|\mathcal{Y}_m)$, where
$\mathcal{Y}_m := \{\bs{Y}_{2\mathcal{C}}, \bs{Y}_{3\mathcal{C}},
\bs{Y}_{3\bar{\mathcal{C}}}\}$,
then inference on
$[\bs{\beta}_{2}|\bs{Y}_{2\bar{\mathcal{C}}},\mathcal{Y}_m]$ would be
straightforward --- at least in principal if not in practice. That is,
suppose
$h(\tilde{\bs{Y}}_{2\bar{\mathcal{C}}}) :=
p(\tilde{\bs{Y}}_{2\bar{\mathcal{C}}}|\mathcal{Y}_m)$,
then one would simply use the likelihood
$h(\bs{Y}_{2\bar{\mathcal{C}}} - \bs{R}_{22}\bs{\beta}_2)$ and a
user-provided prior $g(\cdot)$ over $\bs{\beta}_2$ to calculate a
posterior and return posterior quantities.

However, to reap the benefits of modularity, we describe a procedure
to fit the overall model \eqref{equation:general.model} in two
discrete steps:\ A factor analysis step using \eqref{equation:matrix}
and a regression step using \eqref{equation:true.second.model}. We
begin with estimating the unwanted variation. Specifically, we suppose
that one first assumes the model \eqref{equation:matrix} where
$\tilde{\bs{Y}}_{2\bar{\mathcal{C}}}$ is \emph{unobserved}. The error
$\bs{E}$ can follow any model a researcher desires. Though, of course,
the rotation leading to \eqref{equation:general.model} was derived by
assuming Gaussian errors with independent rows (Section
\ref{subsection:rotate}) and the appropriateness of different error
models should be examined before use. We make the relatively weak
assumption that model \eqref{equation:matrix} is fit using any
Bayesian procedure that yields a proper posterior and that the
researcher can obtain samples from the posterior distribution
$[\tilde{\bs{Y}}_{2\bar{\mathcal{C}}}|\mathcal{Y}_m]$.  Call these
posterior draws
$\tilde{\bs{Y}}_{2\bar{\mathcal{C}}}^{(1)}, \ldots,
\tilde{\bs{Y}}_{2\bar{\mathcal{C}}}^{(t)}$.

After estimating the unwanted variation, we have a regression step
where we estimate $\bs{\beta}_2$ using
\eqref{equation:true.second.model}.  Suppose we have any user-provided
prior density over $\bs{\beta}_2$, say $g(\cdot)$. In order to reap
the benefits of modularity, we need to derive a Bayesian procedure for
approximating the posterior over $\bs{\beta}_2$ using just the samples
$\tilde{\bs{Y}}_{2\bar{\mathcal{C}}}^{(1)}, \ldots,
\tilde{\bs{Y}}_{2\bar{\mathcal{C}}}^{(t)}$
from the previous step. To do so, we let
$\hat{\bs{\beta}}_2^{(i)} :=
\bs{R}_{22}^{-1}(\bs{Y}_{2\bar{\mathcal{C}}} -
\tilde{\bs{Y}}_{2\bar{\mathcal{C}}}^{(i)})$
and note that
$\hat{\bs{\beta}}_2^{(1)},\ldots,\hat{\bs{\beta}}_2^{(t)}$ are
actually draws from
$[\bs{\beta}_2|\bs{Y}_{2\bar{\mathcal{C}}}, \mathcal{Y}_m]$ when using
an (improper) uniform prior over $\bs{\beta}_2$. This follows because
\eqref{equation:true.second.model} is a location family. We can then
weight these samples by the prior information
$g(\hat{\bs{\beta}}_2^{(i)})$ to obtain draws from the posterior
$[\bs{\beta}_2|\bs{Y}_{2\bar{\mathcal{C}}}, \mathcal{Y}_m]$ when using
$g(\cdot)$ as our prior density. This strategy of weighting samples
from one distribution to approximate quantities from another
distribution was discussed in \citet{trotter1956conditional}. What
this means in practice is that given any function of $\bs{\beta}_2$,
say $f(\cdot)$, we can approximate its posterior expectation
consistently in the number of posterior draws, $t$, from the first
step. That is,
\begin{align}
  \label{eq:post.consist}
  \frac{\sum_{i = 1}^t g(\hat{\bs{\beta}}_2^{(i)})f(\hat{\bs{\beta}}_2^{(i)})}{\sum_{i = 1}^t g(\hat{\bs{\beta}}_2^{(i)})} \overset{P}{\longrightarrow} E[f(\bs{\beta}_2)|\bs{Y}_{2\bar{\mathcal{C}}}, \mathcal{Y}_m]
\end{align}
Some example calculations of useful posterior quantities are provided
in Appendix \ref{section:post.summ}. We have given intuitive arguments
here; formal arguments are given in Appendix
\ref{section:just.like}. A technical condition required for this
approach to work is that $g(\cdot)$ be absolutely continuous with
respect to the distribution of
$[\bs{R}_{22}^{-1}(\bs{Y}_{2\bar{\mathcal{C}}} -
\tilde{\bs{Y}}_{2\bar{\mathcal{C}}})|\mathcal{Y}_m]$.
In the case when the errors $\bs{E}$ are Gaussian, it suffices to
consider priors that are absolutely continuous with respect to
Lebesgue measure.

Importantly, this two-step approach, though modular, is actually
fitting the full model \eqref{equation:general.model} as if done in
one step. That is, nothing is lost in taking this two-step approach,
except perhaps we could have found more efficient posterior
approximations if the procedure was fit in one step. However, our
approach is a contrast to RUV2, RUV3, and RUV4 which do not propagate
the uncertainty in estimating the unwanted variation. This allows us
to give more accurate quantities of uncertainty (Section
\ref{section:sims}).

To implement this approach in practice, we need to specify both a
specific model for the unwanted variation \eqref{equation:matrix} and
a prior for $\bs{\beta}_2$. As a proof of concept we use a very simple
Bayesian factor model with Gaussian errors (Appendix
\ref{section:bfa}) and an improper uniform prior on $\bs{\beta}_2$,
which yields a proper proper posterior no matter the model for the
unwanted variation (Appendix \ref{section:proper.posterior}).  We note
that although our model for the unwanted variation is based on a
factor model, RUVB is neither a version of RUV4 nor RUV2.

\subsection{Justification for Approximate Posterior Inference}
\label{section:just.like}
In this section, we prove a general result that given a location
family, we can approximate posterior expectations to any arbitrary
level of precision using samples from the error distribution. We then
connect this to the posterior approximation discussed in Appendix
\ref{section:approx.post.inf}. For data $\bs{y} \in \mathbb{R}^d$,
suppose the model is
\begin{align}
\bs{y} = h(\bs{\theta}) + \bs{e},
\end{align}
where $h: \mathbb{R}^d \rightarrow \mathbb{R}^d$ is bijective.  Let
$\bs{J}(\bs{z})$ be the Jacobian matrix of $h$. That is
\begin{align}
  [\bs{J}(\bs{z})]_{ij} = \frac{\dif h_i(\bs{z})}{\dif z_j},
\end{align}
and let $|\bs{J}(\bs{z})|$ denote its determinant. Let $g$ be the prior
of $\bs{\theta}$, which we assume is absolutely continuous with
respect to the density of $h^{-1}(\bs{y} - \bs{e})$.
\begin{theorem}
  Let $\bs{e}_1,\ldots,\bs{e}_K$ be i.i.d. random variables equal in
  distribution to $\bs{e}$. Let
  $u:\mathbb{R}^d\rightarrow\mathbb{R}^d$ be a function. Let
  \begin{align}
    \label{eq:e.hat}
    \hat{E}[u(\bs{\theta})|\bs{y}] := \frac{\sum_{k = 1}^K u(h^{-1}(\bs{y}-\bs{e}_k)) g(h^{-1}(\bs{y}-\bs{e}_k)) / |\bs{J}(h^{-1}(\bs{y} - \bs{e}_k))|}{\sum_{k = 1}^K g(h^{-1}(\bs{y}-\bs{e}_k)) / |\bs{J}(h^{-1}(\bs{y} - \bs{e}_k))|},
  \end{align}
  then
  \begin{align}
    \hat{E}[u(\bs{\theta})|\bs{y}] \overset{P}{\longrightarrow} E[u(\bs{\theta})|\bs{y}].
  \end{align}
\end{theorem}
\begin{proof}
  Let $f$ be the density of $\bs{e}$. Then $p(\bs{y}|\bs{\theta}) =
  f(\bs{y} - h(\bs{\theta}))$ since $\bs{y}$ belongs to a location
  family. We have
  \begin{align}
     \hat{E}[u(\bs{\theta})|\bs{y}] &:= \frac{\frac{1}{K}\sum_{k = 1}^K u(h^{-1}(\bs{y}-\bs{e}_k))g(h^{-1}(\bs{y}-\bs{e}_k)) / |\bs{J}(h^{-1}(\bs{y} - \bs{e}_k))|}{\frac{1}{K}\sum_{k = 1}^K g(h^{-1}(\bs{y}-\bs{e}_k)) / |\bs{J}(h^{-1}(\bs{y} - \bs{e}_k))|}\\
    \label{equation:wlln}&\overset{P}{\longrightarrow} \frac{E\left[u(h^{-1}(\bs{y}-\bs{e})) g(h^{-1}(\bs{y} - \bs{e}))/ |\bs{J}(h^{-1}(\bs{y} - \bs{e}))|\right]}{E\left[g(h^{-1}(\bs{y} - \bs{e}))/ |\bs{J}(h^{-1}(\bs{y} - \bs{e}))|\right]}\\
    &= \frac{\int u(h^{-1}(\bs{y}-\bs{e}))g(h^{-1}(\bs{y} - \bs{e}))/ |\bs{J}(h^{-1}(\bs{y} - \bs{e}))f(\bs{e}) \dif \bs{e}}{\int g(h^{-1}(\bs{y} - \bs{e}))/ |\bs{J}(h^{-1}(\bs{y} - \bs{e}))f(\bs{e}) \dif \bs{e}}\\
    \label{equation:e.to.z}&= \frac{\int u(\bs{z})g(\bs{z})f(\bs{y} - h(\bs{z})) \dif \bs{z}}{\int g(\bs{z})f(\bs{y} - h(\bs{z})) \dif \bs{z}}\\
    \label{equation:start.post}&= \frac{\int u(\bs{z})p(\bs{z}|\bs{y}) p(\bs{y}) \dif \bs{z}}{p(\bs{y})}\\
    &= \int u(\bs{z})p(\bs{z}|\bs{y}) \dif \bs{z}\\
    &= E[u(\bs{\theta})|\bs{y}].
  \end{align}
  Line \eqref{equation:wlln} follows by two applications of the weak
  law of large numbers followed by Slutsky's theorem. Line
  \eqref{equation:e.to.z} follows by a change of variables
  $\bs{e} =\bs{y} - h(\bs{z})$, the Jacobian of which is just
  $\bs{J}(\bs{z})$. The condition on the prior $g$ is used in
  \eqref{equation:start.post} when we start considering $\bs{z}$ as a
  dummy variable for $\bs{\theta}$. To think intuitively about this
  condition on the prior, if the measure for $\bs{\theta}$ is non-zero
  on a set $\mathcal{A}$ in the parameterspace where the likelihood is
  non-zero but the measure is zero, then this approximation procedure
  would never sample $h^{-1}(\bs{y} - \bs{e}_k) \in \mathcal{A}$. This
  is even though $\mathcal{A}$ does have a non-zero posterior
  probability. The absolute continuity condition prohibits this
  behavior.
\end{proof}

We now connect this general result to Appendix
\ref{section:approx.post.inf}. The $\bs{y}$, $\bs{\theta}$, and
$\bs{e}$ in this section are the $\bs{Y}_{2\bar{\mathcal{C}}}$,
$\bs{\beta}_2$, and $\tilde{\bs{Y}}_{2\bar{\mathcal{C}}}$,
respectively, in Appendix \ref{section:approx.post.inf}. So instead of
having draws $\bs{e}_1,\ldots, \bs{e}_K$, we have that
$\tilde{\bs{Y}}_{2\bar{\mathcal{C}}}^{(1)}, \ldots,
\tilde{\bs{Y}}_{2\bar{\mathcal{C}}}^{(t)}$
are draws from $[\tilde{\bs{Y}}_{2\bar{\mathcal{C}}}|\mathcal{Y}_m]$.
We also have that $h(\bs{\beta}_2) = \bs{R}_{22}\bs{\beta}_2$, and so
the determinant of the Jacobian is merely $|\bs{R}_{22}|^p$. Since
this is a constant independent of $\bs{\beta}_2$, the Jacobians in the
numerator and denominator cancel in \eqref{eq:e.hat}.

\subsection{Simple Bayesian Factor Analysis}
\label{section:bfa}
In this section, we present a simple Bayesian factor analysis which we
used in our implementation of RUVB. The factor model is
\begin{align}
\bs{Y}_{n\times p} = \bs{L}_{n \times q}\bs{F}_{q \times p} + \bs{E}_{n \times p},\
\bs{E} \sim N_{n \times p}(\bs{0}, \bs{\Sigma}\otimes \bs{I}_n),\
\bs{\Sigma}^{-1} = \diag(\xi_1,\ldots,\xi_p).
\end{align}
We use conditionally conjugate priors on all parameters:
\begin{align}
[\bs{L} | \bs{\Psi}] &\sim N_{n \times q}(\bs{0}, \bs{\Psi} \otimes \bs{I}_n),\\
[\bs{F} | \bs{\Sigma}] &\sim N_{q \times p}(\bs{0}, \bs{\Sigma} \otimes \bs{I}_q),\\
[\xi_i|\phi] &\sim \text{Gamma}(\rho_0 / 2, \rho_0\phi / 2),\\
\phi &\sim \text{Gamma}(\alpha_0 / 2, \alpha_0\beta_0 / 2),\\
\bs{\Psi} &= \diag(\zeta_1^{-1},\ldots,\zeta_q^{-1}),\\
\zeta_i &\sim \Gamma(\eta_0 / 2, \eta_0\tau_0/2),
\end{align}
where all hyper-parameters with a $0$ subscript are assumed known. We
let $\text{Gamma}(a, b)$ denote the Gamma distribution with mean $a /
b$ and variance $a / b^2$. In the empirical evaluations of Section
\ref{section:evaluate}, we set the prior ``sample sizes'' to be small
($\rho_0 = \alpha_0 = 0.1$) so that the prior is only weakly
informative. The prior mean for the precisions ($\beta_0$) was set
arbitrarily to 1. Following \citet{ghosh2009default}, we set $\eta_0 =
\tau_0 = 1$.

The prior we use is similar in flavor to that of
\citet{ghosh2009default}, and like them we use the parameter expansion
from \citep{gelman2006prior}, which improves MCMC mixing. However,
there are some important differences between our formulation and that
of \citet{ghosh2009default}. First, we chose not to impose the usual
identifiability conditions on the $\bs{F}$ matrix as we are not
interested in the actual factors or factor loadings. Rather, our
specifications induce a prior over the joint matrix of interest
$\bs{L}\bs{F}$, which is identified. Second, the prior specification
in \citet{ghosh2009default} is not order invariant. That is, their
prior is influenced by the arbitrary ordering of the columns of
$\bs{Y}$. This is a known problem \citep{leung2016order} and we
circumvent it by specifying an order invariant prior. Third, our prior
specifications include hierarchical moderation of the variances, an
important and well-established strategy in gene-expression studies
\citep{smyth2004linear}. Finally, we chose to link the variances of
the genes with those of the factors. This is approximately modeling a
mean-variance relationship and we have found it to work well in
practice. We emphasize here that we do not actually know the form of
the unwanted variation in the simulations in Section
\ref{section:sims}, and so the good performance of RUVB there is not
due to some ``unfair advantage'' in the choice of prior.

One step of a Gibbs sampler is presented in Procedure
\ref{algorithm:bfl.gibbs}. Repeated applications of the steps in
Procedure \ref{algorithm:bfl.gibbs} will result in a Markov chain
whose stationary distribution is the posterior distribution of the
parameters in our model. As the calculations of the full conditionals
for all parameters are fairly standard, we omit the detailed
derivations.

\begin{algorithm}
  \floatname{algorithm}{Procedure}
  \caption{One step of Gibbs sampler for Bayesian factor analysis.}
  \label{algorithm:bfl.gibbs}
  \begin{algorithmic}[1]
    \STATE sample $[\bs{L} | \bs{Y}, \bs{F}, \bs{\Sigma}] \sim N_{n \times q}\left[\bs{Y}\bs{\Sigma}^{-1}\bs{F}^{T}(\bs{F}\bs{\Sigma}^{-1}\bs{F}^{\intercal} + \bs{\Psi}^{-1})^{-1},\ \bs{I}_n \otimes (\bs{F}\bs{\Sigma}^{-1}\bs{F}^{\intercal} + \bs{\Psi}^{-1})^{-1}\right]$,
    \STATE sample $[\bs{F} | \bs{Y}, \bs{L}, \bs{\Sigma}] \sim N_{q \times p}\left[(\bs{L}^{\intercal}\bs{L} + \bs{I}_q)^{-1}\bs{L}^{\intercal}\bs{Y},\ \bs{\Sigma} \otimes (\bs{L}^{\intercal}\bs{L} + \bs{I}_q)^{-1}\right]$,
    \STATE sample $[\xi_i| \bs{Y}, \bs{F}, \bs{L}, \phi] \sim \text{Gamma}\left[(n + q + \rho_0)/2, (r_i + u_i + \rho_0\phi) / 2\right]$, where\\
    $\bs{r} = \diag\left[(\bs{Y} - \bs{L}\bs{F})^{\intercal}(\bs{Y} - \bs{L}\bs{F})\right]$ and $\bs{u} = \diag\left[\bs{F}^{\intercal}\bs{F}\right]$,
    \STATE sample $[\phi|\bs{\xi}] \sim \text{Gamma}\left[(p\rho_0 + \alpha_0) / 2, (\alpha_0\beta_0 + \rho_0 \sum_{i = 1}^p\xi_i)/2\right]$,
    \STATE sample $[\zeta_i | \bs{L}] \sim \text{Gamma}\left[(n + \eta_0) / 2, (s_i + \eta_0\tau_0)/2\right], \text{ where } \bs{s} = \diag(\bs{L}^{\intercal}\bs{L}).$
  \end{algorithmic}
\end{algorithm}

\subsection{Propriety of Posterior}
\label{section:proper.posterior}

In this section, we consider the propriety of the posterior
$\bs{\beta}_2$ from Appendix \ref{section:approx.post.inf}. To prove
that the posterior of $\bs{\beta}_2$ under a uniform prior is proper,
it suffices to consider the model \eqref{equation:true.second.model},
which we repeat here:
\begin{align}
\bs{R}_{22}^{-1}\bs{Y}_{\bar{\mathcal{C}}} = \bs{\beta}_2 + \bs{R}_{22}^{-1}\tilde{\bs{Y}}_{2\bar{\mathcal{C}}}.
\end{align}
Letting $\bs{A} := \bs{R}_{22}^{-1}\bs{Y}_{\bar{\mathcal{C}}}$ and $\bs{C} = \bs{R}_{22}^{-1}\tilde{\bs{Y}}_{2\bar{\mathcal{C}}}$, we have
\begin{align}
\label{equation:location.fam}
\bs{A} = \bs{\beta}_2 + \bs{C}.
\end{align}
Equation \eqref{equation:location.fam} represents a general location
family with location parameter $\bs{\beta}_2$. It is not
difficult to prove that using a uniform prior on the location
parameter in a location family always results in a proper posterior.

\begin{theorem}
Let $\bs{C}$ be a random variable with density $f$.
Let $\bs{A} = \bs{\beta}_2 + \bs{C}$. Suppose we place a uniform prior on
$\bs{\beta}_2$. Then the posterior of $\bs{\beta}_2$ is proper.
\end{theorem}
\begin{proof}
We first note that
\begin{align}
\int f(\bs{C}) \dif \bs{C} = 1.
\end{align}
The density of $\bs{A}$ is just $f(\bs{A} - \bs{\beta}_2)$. Since the prior of $\bs{\beta}_2$ is uniform, this means that the posterior of $\bs{\beta}_2$ given $\bs{A}$ is proportional to the likelihood $f(\bs{A} - \bs{\beta}_2)$. Hence, making a change of variables of $\bs{\eta} = \bs{A} - \bs{\beta}_2$, we have
\begin{align}
\int f(\bs{A} - \bs{\beta}_2) \dif \bs{\beta}_2 = \int f(\bs{\eta}) \dif \bs{\eta} = 1.
\end{align}
\end{proof}

\subsection{Posterior Summaries}
\label{section:post.summ}
Using \eqref{eq:post.consist}, one can obtain approximations to
posterior summaries quite easily. Let $\hat{\bs{\beta}}_2^{(i)} :=
\bs{R}_{22}^{-1}(\bs{Y}_{2\bar{\mathcal{C}}} -
\tilde{\bs{Y}}_{2\bar{\mathcal{C}}}^{(i)})$. The posterior mean may be
approximated by
\begin{align}
  E[\bs{\beta}_2 | \bs{Y}_{2\bar{\mathcal{C}}}, \mathcal{Y}_m]
  \approx \frac{\sum_{i = 1}^t \hat{\bs{\beta}}_2^{(i)}
    g\left(\hat{\bs{\beta}}_2^{(i)} \right)}{\sum_{i =
      1}^tg\left(\hat{\bs{\beta}}_2^{(i)}\right)}.
\end{align}
Suppose one desires a $(1 - \alpha)$ credible interval for
$\beta_{2jk}$ for some $0 < \alpha < 1$. Sort the
$\hat{\beta}_{2jk}^{(i)}$'s such that
$\hat{\beta}_{2jk}^{(1)} < \hat{\beta}_{2jk}^{(2)} < \cdots <
\hat{\beta}_{2jk}^{(t)}$.
A $(1 - \alpha)$ credible interval may be approximated by finding the
$\ell, m \in \{1,\ldots, t\}$ such that
\begin{align}
\label{equation:get.bounds}
\frac{\sum_{i = 1}^{\ell - 1} g\left(\hat{\bs{\beta}}_2^{(i)} \right)}{\sum_{i =
      1}^tg\left(\hat{\bs{\beta}}_2^{(i)}\right)} \leq \alpha / 2 \text{ and }
  \frac{\sum_{i = m + 1}^t g\left(\hat{\bs{\beta}}_2^{(i)} \right)}{\sum_{i =
      1}^tg\left(\hat{\bs{\beta}}_2^{(i)}\right)} \leq \alpha / 2,
\end{align}
then setting the $(1 - \alpha)$ interval as
$(\hat{\beta}_{2jk}^{(\ell)}, \hat{\beta}_{2jk}^{(m)})$. Note that in
\eqref{equation:get.bounds} we have assumed that the
$\hat{\bs{\beta}}_2^{(i)}$'s all have completely distinct elements. If
the error distribution is Gaussian then we may do this without loss of
generality as $\tilde{\bs{Y}}_{2\bar{\mathcal{C}}}^{(i)}$ is drawn
from some convolution with a normal, and so is absolutely continuous
with respect to Lebesgue measure.

Local false sign rates (lfsr's) \citep{stephens2016false} have
recently been proposed as a way to measure the confidence in the sign
of each effect. The intuition is that the lfsr is the probability of
making an error if one makes their best guess about the sign of a
parameter. Let
\begin{align}
p_{jk} := \frac{\sum_{\{i: \hat{\beta}_{2jk}^{(i)} < 0\}} g(\hat{\beta}_{2jk}^{(i)})}{\sum_{i = 1}^t g(\hat{\beta}_{2jk}^{(i)})}.
\end{align}
Then the lfsr's may be approximated by
\begin{align}
  \label{equation:lfsr.def}
  \text{\emph{lfsr}}_{jk} := \min(p_{jk}, 1-p_{jk}),
\end{align}
which simplifies under a uniform prior to
\begin{align}
  \text{\emph{lfsr}}_{jk} := \frac{1}{t}\min\left[ \#\{\hat{\beta}_{2jk}^{(i)} < 0\},
    \#\{\hat{\beta}_{2jk}^{(i)} > 0\}\right].
\end{align}
Though, in many cases in Section \ref{section:evaluate}, the Markov
chain during the Bayesian factor analysis did not sample
$\beta_{2jk}$'s of opposite sign. The estimate for the lfsr using
\eqref{equation:lfsr.def} would then be 0. As it is often desirable to
obtain a ranking of the most significant genes, this is an unappealing
feature. We instead use a normal approximation to estimate the the
lfsr's, using the posterior means and standard deviations from the
samples of the $\beta_{2jk}$'s.

\subsection{Additional Simulation Considerations}
\label{section:additional.considerations}

To implement the methods in the simulation studies in Section
\ref{section:sims}, there are additional technicalities to be
considered. Here, we briefly list out our choices.

The number of factors, $q$, is required to be known for all methods
that we examined. There are many approaches to estimate this in the
literature \citep[for a review, see][]{owen2016bi}. We chose to use
the parallel analysis approach of \citep{buja1992remarks} as
implemented in the \texttt{num.sv} function in the R package
\texttt{sva} \citep{leek2016sva}. An alternative choice could have
been the bi-cross-validation approach described in \citep{owen2016bi}
and implemented in the \texttt{cate} R package \citep{wang2015cate}.

RUV2, RUV3, RUV4, and CATE all require specifying a factor analysis
(Definition \ref{def:fa}) for their respective first steps. For all
methods, we used the truncated SVD
\eqref{equation:tsvd.z}-\eqref{equation:tsvd.sigma}. This was to make
the methods more comparable. Though we note that
\citet{wang2015confounder} also suggest using the quasi-maximum
likelihood approach of \citet{bai2012statistical}.

For RUVB, we ran each Gibbs sampler (Algorithm
\ref{algorithm:bfl.gibbs}) for 12,500 iterations, dropping the first
2,500 iterations as burn-in. We kept every 10th sample to retain 1000
posterior draws from which we calculated the posterior summaries of
Appendix \ref{section:post.summ}. Convergence diagnostics and other
checks were implemented on a sample of the Markov chains in our
simulation studies. We detected no problems with convergence (data not
shown).

Since we explored many combinations of methods, we adopt the following notation:
\begin{itemize}[noitemsep, nolistsep]
\item o = original variance estimates,
\item m = MAD variance calibration,
\item c = control-gene variance calibration,
\item l = limma-moderated variances (EBVM),
\item lb = limma-moderated variances (EBVM) before GLS (for either CATE or RUV3),
\item la = limma-moderated variances (EBVM) after GLS (for either CATE or RUV3),
\item d = delta-adjustment from CATE package (additive variance inflation),
\item n = $t$ approximation for the likelihood in RUVB,
\item nn = normal approximation for the likelihood in RUVB.
\end{itemize}
When comparing the AUC of different methods, certain combinations of
methods theoretically have the same AUC. Specifically, applying MAD
variance calibration (m) or control-gene variance calibration (c) does
not alter the AUC of a method. Thus, we only need to compare one
method of each of these groups to obtain comprehensive results on AUC
performance. The members of these groups are those shown in
Supplementary Figure \ref{figure:diff_full}.

The effect of library size (the total number of gene-reads in a
sample) is a well-known source of bias in RNA-seq data
\citep{dillies2013comprehensive}. We do not explicitly adjust for
library size. In this paragraph, we briefly argue that RUV-like
methods can effectively account for library size. The usual pipeline
to adjust for library size is to choose a constant for each sample,
$c_i$, and divide each count in a sample by $c_i$. Many proposals have
been made to estimate $c_i$ \citep{anders2010differential,
  bullard2010evaluation, robinson2010scaling}. There are also
variations on this pipeline. For example, others choose a constant for
each sample, $c_i$, and include the
$\bs{c} = (c_1,\ldots,c_n)^{\intercal}$ vector as a covariate in the
regression model \citep{langmead2010cloud}. In terms of our
$\log_2$-count matrix of responses $\bs{Y}$, this corresponds to
fitting the model
\begin{align}
\label{equation:library.size}
\bs{Y} = \bs{X}\bs{\beta} +
\bs{c}\bs{d}^{\intercal} + \bs{E},
\end{align}
where $\bs{c}$ is estimated independently of $\bs{X}$
in some (ad-hoc) fashion and $\bs{d}$ may or may not be
assumed to be the vector of ones, $\bs{1}_p$. However,
equation \eqref{equation:library.size} is just a factor-augmented
regression model, the same as \eqref{equation:full.model}. Methods
that assume model \eqref{equation:full.model} thus need not adjust for
library size and need not choose one of the many procedures to
estimate $\bs{c}$. That is,
library size can be treated as just another source of unwanted variation.

\subsection{Calibrated CATE}
\label{section:catec}

We can derive a multivariate version of \eqref{eq:gb.var.inflate} and
formally justify it with maximum likelihood arguments via a
modification of step 2 of Procedure \ref{algorithm:ruv4}. After step
1, we modify \eqref{equation:control.model} and
\eqref{equation:control.model.var} to include a variance inflation
parameter, $\lambda$.
\begin{align}
  \bs{Y}_{2\mathcal{C}} &= \bs{Z}_2\hat{\bs{\alpha}}_{\mathcal{C}} +
  \bs{E}_{2\mathcal{C}}, \label{equation:control.model.mod}\\
  e_{2\mathcal{C}ij} &\overset{ind}{\sim} N(0,
  \lambda\hat{\sigma}_j^2). \label{equation:control.model.var.mod}
\end{align}
Step 2 of CATE simply calculates the MLE of $\bs{Z}_2$ in
\eqref{equation:control.model} and
\eqref{equation:control.model.var}. Our calibration is simply finding
the MLE's of $\bs{Z}_2$ and $\lambda$ in
\eqref{equation:control.model.mod} and
\eqref{equation:control.model.var.mod}, which can be done in closed
form. The estimate of $\bs{Z}_2$ is unchanged
\eqref{equation:cate.z2.est}. The MLE of $\lambda$ is
\begin{align}
\label{equation:lambda.mle}
\hat{\lambda} = \frac{1}{k_2m}\tr\left[(\bs{Y}_{2\mathcal{C}} - \hat{\bs{Z}}_2\hat{\bs{\alpha}}_{\mathcal{C}})\hat{\bs{\Sigma}}_{\mathcal{C}}^{-1}(\bs{Y}_{2\mathcal{C}} - \hat{\bs{Z}}_2\hat{\bs{\alpha}}_{\mathcal{C}})^{\intercal}\right].
\end{align}
Inference then proceeds by using $\hat{\lambda}\hat{\sigma}_j$ as the
estimate of $\sigma_j$. Formula \eqref{equation:lambda.mle} is a
multivariate version of equation (236) of \citet{gagnon2013removing}
(and \eqref{eq:gb.var.inflate}). Though while we derived
\eqref{equation:lambda.mle} using an application of maximum
likelihood, \citet{gagnon2013removing} formulated their calibration
using heuristic arguments.

\bibliography{vicar_bib}

\clearpage

\section*{}
\label{section:supp}

\huge
\vspace{4in}
\begin{center}
Supplementary Materials
\end{center}
\normalsize

\begin{suppfigure}
\begin{center}
\includegraphics{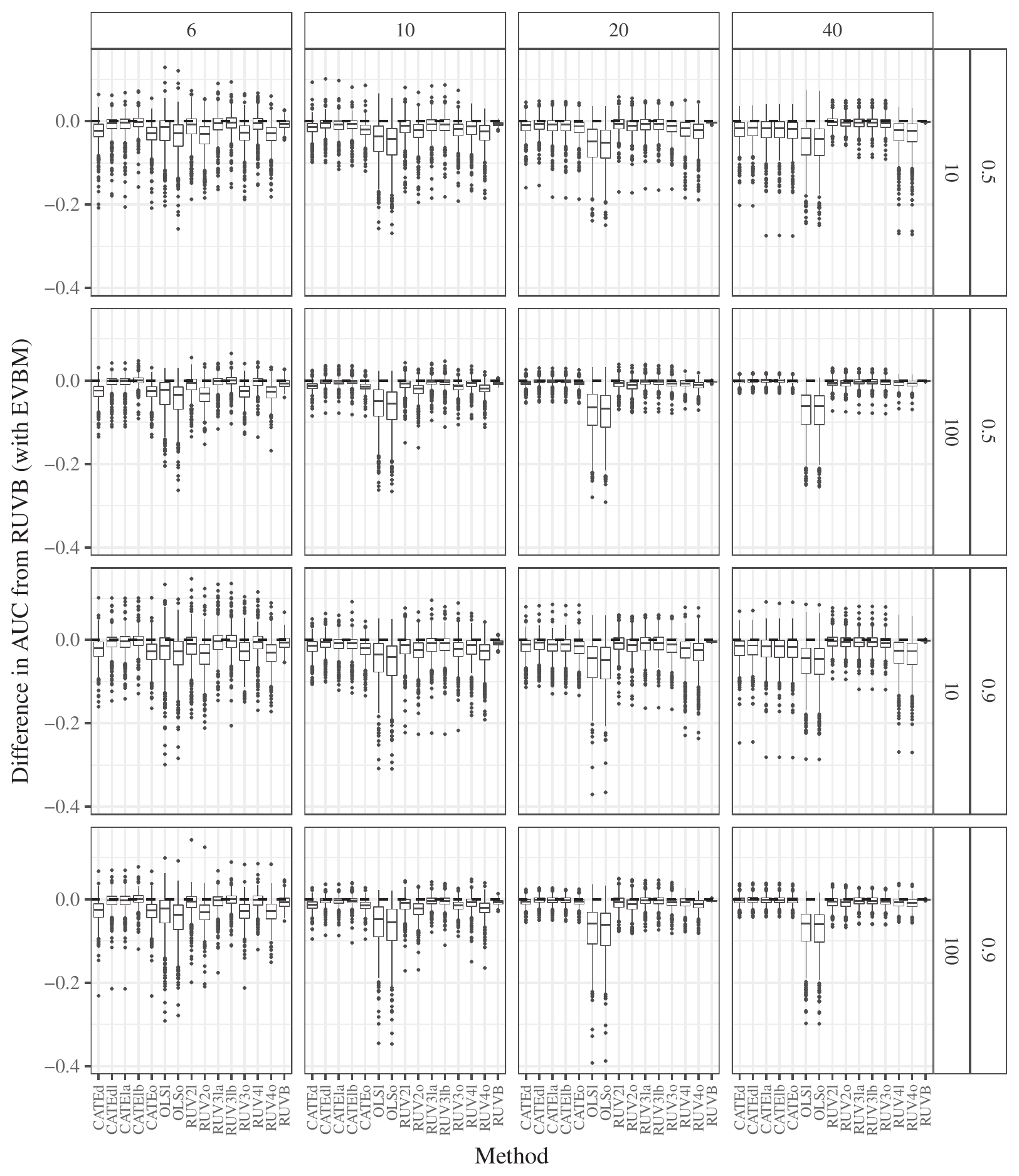}
\caption{Boxplots of AUC of methods subtracted from the AUC of RUVB
  (with EBVM). Anything above zero (the dashed horizontal line)
  indicates superior performance to RUVB. Anything below the line
  indicates inferior performance to RUVB. Columns are the sample
  sizes, rows are the number of control genes by the proportion of
  genes that are null. For the notation of the methods, see Appendix
  \ref{section:additional.considerations}.}
\label{figure:diff_full}
\end{center}
\end{suppfigure}

\begin{suppfigure}
\begin{center}
\includegraphics{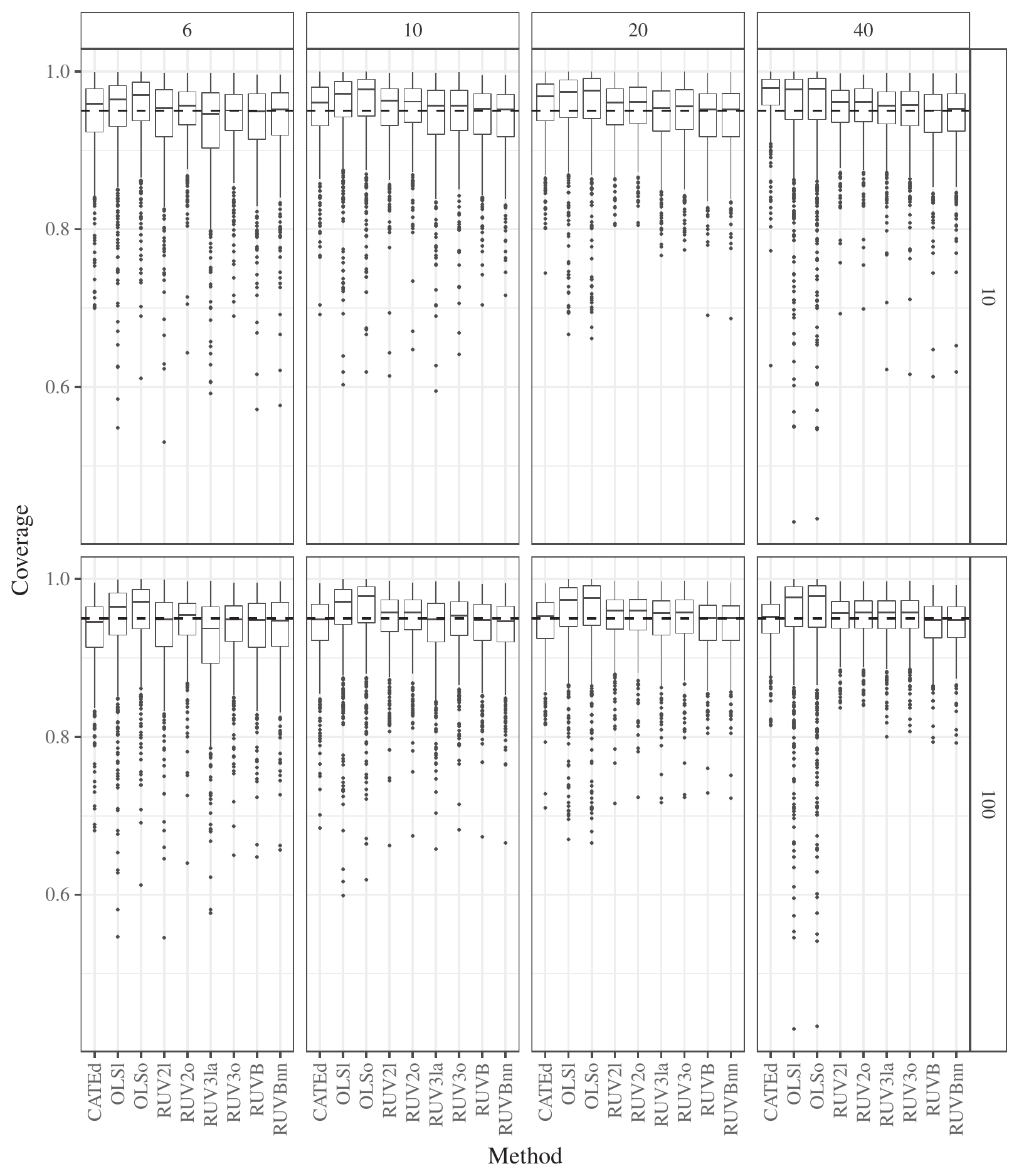}
\caption{Boxplots for the coverage of the 95\% confidence intervals
  for the best-performing methods when $\pi_0 = 0.5$. The column
  facets index the sample sizes used while the horizontal facets index
  the number of control genes used. The horizontal dashed line is at
  0.95. For the notation of the methods, see Appendix
  \ref{section:additional.considerations}.}
\label{fig:best.boxplots.coverage}
\end{center}
\end{suppfigure}

\end{document}